\documentclass[sigplan,nonacm]{acmart}

\renewcommand\footnotetextcopyrightpermission[1]{} % removes footnote with conference information in first column
\renewcommand\footnotetextcopyrightpermission[1]{}

\makeatletter
\renewcommand\@formatdoi[1]{\ignorespaces}
\makeatother
\footnotetextcopyrightpermission

\acmSubmissionID{688}
\setcopyright{none}

\usepackage{amsmath,amsfonts,amsthm}
\usepackage{algorithm}
\usepackage[noend]{algpseudocode}
\usepackage[shortlabels]{enumitem}
\usepackage{graphicx}
\usepackage{textcomp}
\usepackage{xcolor}
\usepackage{url}
\usepackage{xspace}
\usepackage{cleveref}
\usepackage{pifont}
\usepackage{booktabs}
\usepackage{wrapfig}
\usepackage[font=footnotesize]{caption}
\usepackage{thmtools}
\usepackage{thm-restate}
\usepackage{hyperref}
\usepackage{tabularx}
\usepackage{enumitem}
\usepackage{caption}
\usepackage{subcaption}

\algdef{SE}[Upon]{Upon}{EndUpon}[1]{\textbf{upon}\ #1\ \algorithmicdo}{\algorithmicend\ \textbf{}}%
\algtext*{EndUpon}

\microtypecontext{spacing=nonfrench}
\usepackage{todonotes}

\def\cameraReady{} % set to true

\begin{document}

% ====== comments ======

\ifdefined\cameraReady
\renewcommand{\todo}[2][]{}
\fi

\newcommand{\textcomment}[2]{\textbf{#1:} #2}
\newcommand\rati[1]{\todo[color=yellow,inline]{\textcomment{Rati}{#1}}}
\newcommand\sasha[1]{\todo[color=red,inline]{\textcomment{Sasha}{#1}}}
\newcommand\zekun[1]{\todo[color=brown,inline]{\textcomment{Zekun}{#1}}}
\newcommand\balaji[1]{\todo[color=purple,inline]{\textcomment{Balaji}{#1}}}

\newcommand\com[1]{}

% ====== systems ======
\newcommand\sysname{Shoal\xspace}

\newcommand\chainspace{Chainspace\xspace}
\newcommand\ethereum{Ethereum\xspace}
\newcommand\hyperledger{Hyperledger\xspace}
\newcommand\omniledger{Omniledger\xspace}
\newcommand\rapidchain{RapidChain\xspace}
\newcommand\coconut{Coconut\xspace}
\newcommand\bft{BFT\xspace}
\newcommand\rscoin{RSCoin\xspace}
\newcommand\bftsmart{\textsc{bft-SMaRt}\xspace}
\newcommand\byzcuit{Byzcuit\xspace}
\newcommand\bitcoin{Bitcoin\xspace}
\newcommand\bitcoinng{Bitcoin-NG\xspace}
\newcommand\cosi{CoSi\xspace}
\newcommand\byzcoin{ByzCoin\xspace}
\newcommand\elastico{Elastico\xspace}
\newcommand\algorand{Algorand\xspace}
\newcommand\hyperledgerfabric{Hyperledger Fabric\xspace}
\newcommand\pbft{PBFT\xspace}
\newcommand\bftlong{Byzantine Fault-Tolerant\xspace} 
\newcommand\solidus{Solidus\xspace}
\newcommand\hashgraph{Hashgraph\xspace}
\newcommand\avalanche{Avalanche\xspace}
\newcommand\blockmania{Blockmania\xspace}
\newcommand\stellar{Stellar\xspace}
\newcommand\libra{Libra\xspace}
\newcommand\librabft{LibraBFT\xspace}
\newcommand\move{Move\xspace}
\newcommand\hotstuff{HotStuff\xspace}
\newcommand\vanillahs{Vanilla-HotStuff\xspace}
\newcommand\batchedhs{Batched-HotStuff\xspace}

%  ===== custom notations ======
\newcommand{\keyword}[1]{\normalfont \texttt{#1}}
\newcommand{\accounts}{\keyword{accounts}}

\newcommand{\transfer}{O}
\newcommand{\cert}{C}
\newcommand{\sync}{S}
\newcommand{\account}{a}
\newcommand{\authority}{\alpha}

%  ===== formatting ======
\newcommand{\para}[1]{\vspace{2mm}\noindent\textbf{#1}.\xspace}

% Abbreviations
\newcommand{\cf}{cf.\@\xspace}
\newcommand{\vs}{vs.\@\xspace}
\newcommand{\etc}{etc.\@\xspace}
\newcommand{\ala}{ala\@\xspace}
\newcommand{\wrt}{w.r.t.\@\xspace}
\newcommand{\etal}{\textit{et al.}\@\xspace}
\newcommand{\eg}{\textit{e.g.}\@\xspace}
\newcommand{\ie}{\textit{i.e.}\@\xspace}
\newcommand{\via}{\textit{via}\@\xspace}
\newcommand{\defacto}{\textit{de facto}\@\xspace}

% Theorems
\newtheorem{assumption}{Security Assumption}
\newtheorem{property}{Property}
%\newtheorem{proposition}{Proposition}
%\newtheorem{lemma}{Lemma}

% For inline section titles
\newcommand\inlinesection[1]{{\bf #1.}}

\def\first{({i})\xspace}
\def\second{({ii})\xspace}
\def\third{({iii})\xspace}
\def\fourth{({iv})\xspace}
\def\fifth{({v})\xspace}
\def\sixth{({vi})\xspace}

\newcommand{\one}{({i})\xspace}
\newcommand{\two}{({ii})\xspace}
\newcommand{\three}{({iii})\xspace}
\newcommand{\four}{({iv})\xspace}
\newcommand{\five}{({v})\xspace}
\newcommand{\six}{({vi})\xspace}

% Colors
\definecolor{verylightgray}{gray}{0.9}

% Markers
\newcommand\vgap{\vskip 2ex}
\newcommand\marker{\vgap\ding{118}\xspace}
\def\na{--}
\def\unsure{?}
\def\missing{$!$}
\newcommand{\yes}{\ding{51}}
\newcommand{\no}{\ding{55}}
\DeclareRobustCommand\pie[1]{
\tikz[every node/.style={inner sep=0,outer sep=0, scale=1.5}]{
\node[minimum size=1.5ex] at (0,-1.5ex) {}; 
\draw[fill=white] (0,-1.5ex) circle (0.75ex); \draw[fill=black] (0.75ex,-1.5ex) arc (0:#1:0.75ex); 
}}
\def\L{\pie{0}} % Low
\def\M{\pie{-180}} % Medium
\def\H{\pie{360}} % High

% Listings

%Algorithmicx
\algdef{SE}[Receiving]{Receiving}{EndReceiving}[1]{\textbf{upon
receiving}\ #1\ \algorithmicdo}{\algorithmicend\ \textbf{}}%
\algtext*{EndReceiving}
\newcommand\StateX{\Statex\hspace{\algorithmicindent}}
\algrenewcommand\textproc{}% Used to be \textsc

\graphicspath{{figures/}}

\date{}

\title{
%\sysname: Reducing DAG-BFT Latency With Pipelining, Leader Reputation and Prevalent Responsiveness
%\sysname: Reducing DAG-BFT Latency 
\sysname: Improving DAG-BFT Latency And Robustness  
}

\ifdefined\cameraReady

\author{Alexander Spiegelman}
%\email{sasha@aptoslabs.com}
\affiliation{
    \institution{Aptos}
    \country{}
    }
%\country{United States}

\author{Balaji Arun}
\affiliation{
    \institution{Aptos}
    \country{}
    }
%\country{United States}

\author{Rati Gelashvili}
\affiliation{
    \institution{Aptos}
    \country{}
    }
%\country{United States}

\author{Zekun Li}
\affiliation{
    \institution{Aptos}
    \country{}
    }
%\country{United States}

\else
\author{}
\fi

\settopmatter{printfolios=false}
\settopmatter{printacmref=true}

\begin{abstract}
The Narwhal system is a state-of-the-art Byzantine fault-tolerant scalable architecture that involves constructing a directed acyclic graph (DAG) of messages among a set of validators in a Blockchain network.
Bullshark is a zero-overhead consensus protocol on top of the Narwhal's DAG that can order over 100k transactions per second.
Unfortunately, the high throughput of Bullshark comes with a latency price due to the DAG construction, increasing the latency compared to the state-of-the-art leader-based BFT consensus protocols.

We introduce \sysname, a protocol-agnostic framework for enhancing Narwhal-based consensus. 
By incorporating leader reputation and pipelining support for the first time, \sysname significantly reduces latency.
Moreover, the combination of properties of the DAG construction and the leader reputation mechanism enables the elimination of timeouts in all but extremely uncommon scenarios in practice, a property we name ``prevalent responsiveness" (it strictly subsumes the established and often desired ``optimistic responsiveness" property for BFT protocols). 

We integrated \sysname instantiated with Bullshark, the fastest existing Narwhal-based consensus protocol, in an open-source Blockchain project and provide experimental evaluations demonstrating up to 40\% latency reduction in the failure-free executions, and up-to 80\% reduction in executions with failures against the vanilla Bullshark implementation.
\end{abstract}




\maketitle

  %  \sasha{set submission ID}

\section{Introduction} 
\label{sec:introduction}

%Change the story like in the blog post.
%Say that Bullshark is great for TP but increses The latency of Jolteon by 50%.
%We bring it down to the same level!

Byzantine fault tolerant (BFT) systems, including consensus protocols~\cite{zyzzyva, castro2002practical, 700bft, sbft} and state machine replication~\cite{bessani2014state, hotstuff, diembft, buchman2018latest, kapitza2012cheapbft}, have been a topic of research for over four decades as a means of constructing reliable distributed systems.
Recently, the advent of Blockchains has underscored the significance of high performance.
While Bitcoin handles approximately 10 transactions per second (TPS), the proof-of-stake committee-based blockchains~\cite{solana, aptos, ava, sui, celo, definity} are now engaged in a race to deliver a scalable BFT system with the utmost throughput and minimal latency.

Historically, the prevailing belief has been that reducing communication complexity was the key to unlocking high performance, leading to the pursuit of protocols with linear communication. However, this did not result in drastic enough improvements in the throughput, falling significantly short of the current blockchain network targets. For example, the state-of-the-art Hotstuff~\cite{hotstuff} protocol in this line of work only achieves a throughput of 3500 TPS~\cite{bottlenecks}.

A recent breakthrough, however, stemmed from the realization that data dissemination is the primary bottleneck for leader-based protocols, and it can benefit from parallelization~\cite{narwhaltusk, MirBFT, dispersedledger, arun2022dqbft}. The Narwhal system~\cite{narwhaltusk} separated data dissemination from the core consensus logic and proposed an architecture where all validators simultaneously disseminate data, while the consensus component orders a smaller amount of metadata. A notable advantage of this architecture is that not only it delivers impressive throughput on a single machine, but also naturally supports scaling out each blockchain validator by adding more machines. The Narwhal paper~\cite{narwhaltusk} evaluated the system in a geo-replicated environment with 50 validators and reported a throughput of 160,000 TPS with one machine per validator, which further increased to 600,000 TPS with 10 machines per validator. 

These numbers are more in line with the ambitions of modern blockchain systems.  Consequently, Narwhal has garnered significant traction within the community, resulting in its deployment in Sui~\cite{sui} and ongoing development in Aptos~\cite{aptos} and Celo~\cite{celo}.

Developing a production-ready reliable distributed system is challenging, and integrating intricate consensus protocols only adds to the difficulty. 
Narwhal addresses this issue by abstracting away networking from the consensus protocol. It constructs a non-equivocating round-based directed acyclic graph (DAG), a concept initially introduced by Aleph~\cite{aleph}. In this design, each validator contributes one vertex per round, and each vertex links to $n-f$ vertices in the preceding round.
Each vertex is disseminated via an efficient reliable broadcast implementation, ensuring that malicious validators cannot distribute different vertices to different validators within the same round.
With networking abstraction separated from the details of consensus, the DAG can be constructed without contending with complex mechanisms like view-change or view-synchronization.

During periods of network asynchrony, each validator may observe a slightly different portion of the DAG at any given time. However, the structure facilitates a simpler ordering mechanism compared to monolithic BFT protocols. In DAG-based consensus protocols, vertices represent proposals, edges represent votes, and the concept of quorum intersection guarantees that validators can consistently order all DAG vertices. This provides efficient consensus because ordering is done via local computation only, without any additional communication cost.

 %enabling the consensus process to progress efficiently and reliably.

% Additionally, the system achieves perfect load balancing. 
% load balancing and garbage collection, is it needed?

% bind into simplified ordering logic.

%This structure offers two significant advantages in practice: (1) straightforward garbage collection and (2) simplified ordering logic. For the latter, 
\paragraph{Narwhal-based consensus protocols}
As discussed, the idea shared by Narwhal-based consensus protocols is to interpret the DAG structure as the consensus logic~\cite{allyouneed, narwhaltusk, bullshark, bullsharksync}, but they differ in the networking assumptions and the number of rounds required for vertex ordering.
However, all three protocols share a common structure.
Prior to the protocol initiation, there is an a-priori mapping from specific rounds to leaders shared among all validators.
In the asynchronous protocols (DAG-Rider and Tusk), this mapping to the sequence of leaders is hidden behind threshold cryptography and revealed throughout the protocol.
%In the partially synchronous Bullshark the sequence is global and known to all.
We use the term \emph{anchor} to refer to the vertex associated with the round leader in each relevant round. 
%It is important to note that, in these protocols, not every round includes an anchor. This has implications for latency, as we will explain shortly.

The DAG local ordering process by each validator is divided into two phases. First, each validator determines which anchors to order (the rest are skipped). Then, the validators sequentially traverse the ordered anchors, deterministically ordering all DAG vertices contained within the causal histories of the respective anchors. The primary considerations that affect the protocol latency are as follows
\begin{enumerate}
    \item Bad leaders. When a validator is malicious or not fast enough, its vertex may not be included in the DAG. 
    In the case of leaders, the absence of anchors affects the ordering latency of all vertices in previous rounds that are not already ordered.
    These vertices can only be ordered as a part of a causal history of a future anchor, directly impacting their latency.

    \item Sparse anchors. In Narwhal-based consensus protocols, not every round includes an anchor. Consequently, vertices located farther from the next anchor must wait for additional rounds before they can be ordered.
    
    %This has implications for latency, as we will explain shortly.

    %As mentioned above, the ordering logic for anchors does not allow for an anchor in every round. Consequently, vertices located farther from the next anchor must wait for additional rounds before they can be ordered.  

\end{enumerate}

\paragraph{\sysname framework}
This paper presents \sysname: a framework addressing the aforementioned challenges by incorporating leader reputation and pipelining mechanisms into all Narwhal-based consensus protocols.
So far, all available open-source implementations of Narwhal and Bullshark, including Meta~\footnote{https://github.com/facebookresearch/narwhal/blob/main/consensus/src/lib.rs}, and the production deployment on Sui~\footnote{https://github.com/MystenLabs/sui/blob/main/narwhal/consensus/src/bullshark} lack these features, while our evaluations demonstrate they can provide significant performance improvements.

Leader reputation is an often overlooked concept in theoretical research, yet it holds crucial importance for practical performance. 
In practice, Byzantine failures are rare due to robust protection and economic incentives for validators to adhere to the protocol. (Moreover, Narwhal-based DAG constructions, which provide non-equivocation, significantly reduce the range of potential Byzantine behavior).  
Thus, the most common failure scenarios in Blockchain (esp. in Narwhal-based) systems involve validators who struggle to keep up, which can occur due to temporary crashes, slower hardware, or geographical distance. If unresponsive validators repeatedly become leaders, progress is inevitably impeded and degrades system performance. The leader reputation schemes select leaders based on the history of their recent activity, as introduced in Diem~\cite{diembft} and later formalized in~\cite{cohen2022aware}. 

In the context of Narwhal-based consensus, pipelining means having an anchor in every round, which would result in improved latency for non-anchor vertices. 

\paragraph{The main challenge}
While the ability to order the DAG locally, without extra communication contributes to the scalability of Narwhal-based consensus, it poses a significant challenge to supporting leader reputation and pipelining. 

The leader reputation problem is simpler to solve for monolithic BFT consensus protocols. While the validators may disagree on the history that determines the next leader's identity, the worst that can happen is a temporary loss of liveness until view synchronization, i.e. the quorum of validators can eventually recover by agreeing on a fall-back leader.
This exact method was utilized in~\cite{cohen2022aware}, electing the fall-back leaders by a simple round-robin.

In contrast, when all communication is done upfront for building the DAG, the safety of a consensus protocol relies on a key property of the local computation that all validators will decide to order the same set of anchors.
This must hold despite the local views of the DAG possibly differing among the validators across multiple rounds.
Hence, selecting the round leaders dynamically based on reputation (as opposed to the a-priori mapping) seems impossible due to a circular dependency: we need to agree on mapping to solve consensus, but we need consensus to agree on a new mapping.

For pipelining, even if all validators agree on the mapping, they also must agree on whether to order or skip each anchor.
Our attempts to solve the problem by delving into the inner workings of the protocol and exploring complex quorum intersection ordering rules have not been fruitful. 
Intuitively, this is because consensus requires a voting round after each anchor proposal and the next anchor should link to the decisions (votes) on the previous one.     

\paragraph{Our solution.}
In \sysname, we lean into the power of performing computations on the DAG, in particular the ability to preserve and re-interpret information from previous rounds. For leader reputation, this allows bootstrapping the seemingly circular dependency on consensus, while for pipelining, it allows combining multiple instances of the protocol in a suitable manner.
In fact, \sysname runs multiple instances of the protocol one after the other, where the trick is to agree on the switching point based on the following observation:\\

\textit{For any Narwhal-based consensus protocol, since all validators agree on which anchors to order vs skip, they in particular agree on the first ordered anchor.}\\

With this observation in mind, 
%we can easily switch between the instances of the protocol.
each validator can start locally interpreting its view of the DAG by running an instance of its favorite protocol until it determines the first ordered anchor.
Since validators agree on this anchor, they can all deterministically start a new protocol instance in the following round.
Note that this too, happens locally, from a validator's perspective, as a part of re-interpreting the DAG.
As a result, \sysname ensures the following
\begin{enumerate}
    \item Leader reputation: validators select new anchors for future rounds based on the information available in the causal history of the ordered anchors. 
    %That is, avoiding validators that did not previously perform from serving as leaders.

    \item Pipelining: allocate an anchor in the first round of the new instance. That way, if the first anchor in every instance is ordered, we get an anchor in every round, providing the pipelining effect.
\end{enumerate}

\paragraph{Our system and prevalent responsiveness}
We implemented \sysname in the open-source codebase of one of the live Blockchain networks and instantiated it with the partially synchronous version of Bullshark\footnote{\sysname of bull sharks.}. 
In this setting, we also discovered a way to eliminate timeouts in all except extremely rare scenarios, a property we refer to as prevalent responsiveness. 
The design with prevalent responsiveness demonstrates further performance improvements in our evaluations.
Added motivation to avoid timeouts in as many situations as possible comes from a purely practical point of view, as (1) when timeouts are common, the duration affects the system performance, but in a way that is non-trivial to configure in an optimal way as it is highly environmentally (network) dependent; and (2) timeout handling is known to add significant complexity to the implementation logic for managing potential state space of validators.

Monolithic leader-based BFT protocols use timeouts to trigger protocol progress every time a leader is faulty or slow, while optimistic responsiveness property, popularized by the HotStuff~\cite{hotstuff} protocol, effectively eliminates timeout implications in ideal scenarios when the network is synchronous and there are no failures.
However, when failures do occur, all validators must still wait until the timeout expires before transitioning to the next leader.

Utilizing the inherent properties of the DAG construction, and leader reputation mechanism, we ensure that \sysname makes progress at network speed under a much larger set of scenarios than optimistically responsive protocols would, which makes \sysname with partially synchronous Bullshark prevalently responsive. In \sysname, validators do wait for timeouts when a few leaders crash and the corresponding anchors are not ordered. 
While the FLP~\cite{flp} impossibility result dictates that there has to be a scenario that requires a timeout, \sysname design aligns this FLP scenario to be extremely improbably in practice (multiple, e.g., 10 consecutive skipped anchors). Conceptually, this is similar to how randomized protocols align FLP scenarios to have $0$ probability in solving asynchronous consensus with probability $1$~\cite{ben1983another}.

All available Bullshark implementations use timeouts to ensure honest validators wait for slow anchors even if $2f+1$ other vertices were already delivered.   
By eliminating timeouts, \sysname immediately reduces latency when a leader is faulty, as the corresponding anchors would never be delivered and it is best to advance to the next round as fast as possible.
If the leader is not crashed and just slower, validators may skip anchors that they could order if they waited a little bit longer.
This is however, where the leader reputation mechanism of \sysname shines, filtering out slow validators that constantly delay new rounds and allowing the DAG to proceed at network speed while ordering most anchors.

Our experimental evaluation demonstrates up to 40\% reduction in latency against vanilla Bullshark protocol implementation when there are no failures in the system, and up to 80\% reduction in latency when there are failures. We provide experiments specifically designed to give insights into the impact of the improvements separately, i.e. pipelining, leader reputation and eliminating the timeouts (prevalent responsiveness).

In summary, the paper focuses on improving latency and robustness in DAG-Based protocols.
It provides \sysname, a framework to enhance any Narwhal-based consensus protocol with (1) Leader reputation mechanism that prevents slow, isolated, or crashed validators from becoming leaders, (2) pipelining support that ensures every round on the DAG has an anchor, and (3) eliminating timeouts in many cases further reducing the latency,
%and improving the robustness of the system.

The remaining sections of the paper are organized as follows:
Section~\ref{sec:dagbft} provides background information on DAG-BFT and highlights the main property utilized in this paper.
Section~\ref{sec:pipelining} introduces our pipelining approach, while Section~\ref{sec:leaderreputation} presents the leader reputation solution in \sysname.
In Section~\ref{sec:security}, we prove correctness of the proposed framework.
Section~\ref{sec:implementation} describes the implementation details and discusses timeouts. Section~\ref{sec:evaluation} presents the results of our evaluation.
Section~\ref{sec:related} discusses related work, and finally, Section~\ref{sec:conclusion} concludes the paper.

\section{DAG BFT} \label{sec:dagbft}
%DAG-based BFT consensus protocols are gaining increasing popularity in Blockchain networks. In fact, Narwhal~\cite{} (Best paper award in EuroSys'22) and Bullshark are already deployed in Sui network~\cite{} and are currently being implemented by Aptos network~\cite{} and Celo~\cite{}.
%\sasha{The above is in the intro now. Maybe we can skip it here.}

%As mentioned in the Introduction, Narwhal offers a highly scalable and efficient DAG construction. DAG-rider~\cite{}, Tusk~\cite{}, and Bullshark~\cite{} are all local consensus algorithms that totally order the nodes of the DAG without incurring any communication overhead. Together, they can achieve an ordering throughput of over 100k TPS in a geo-replicated environment~\cite{Narwhal, Bullshark}.

We start by providing the necessary background on Narwhal-based BFT consensus (Section \ref{sub:background}) and define a common property (Section \ref{sub:framework}) satisfied by such consensus protocols. We rely on this property while designing \sysname to enhance a given baseline protocol with pipelining and leader reputation, thereby reducing latency.

\subsection{Background}
\label{sub:background}
The concept of DAG-based BFT consensus, initially introduced by HashGraph~\cite{hashgraph}, aims to decouple the network communication layer from the consensus logic. In this approach, each message consists of a collection of transactions and references to previous messages. These messages collectively form an ever-growing DAG, with messaging serving as vertices and references between messages serving as edges.
%Due to the asynchronous nature of the network, different validators may see slightly different local views of the DAG at any given time.

%\subsubsection{Narwhal's DAG}
In Narwhal, the DAG is round-based, similar to Aleph~\cite{aleph}. In this approach, each vertex within the DAG is associated with a round number. 
In order to progress to round $r$, a validator must first obtain $n-f$ vertices (from distinct validators) belonging to round $r-1$.
Every validator can broadcast one vertex per round, with each vertex referencing a minimum of $n-f$ vertices from the previous round.

The \emph{causal history} of a vertex $\mathsf{v}$ refers to the sub-graph that starts from $\mathsf{v}$. Figure~\ref{fig:dag} illustrates a validator's local view of a round-based DAG.

\begin{figure}
    \centering
    \includegraphics[width=0.48\textwidth]{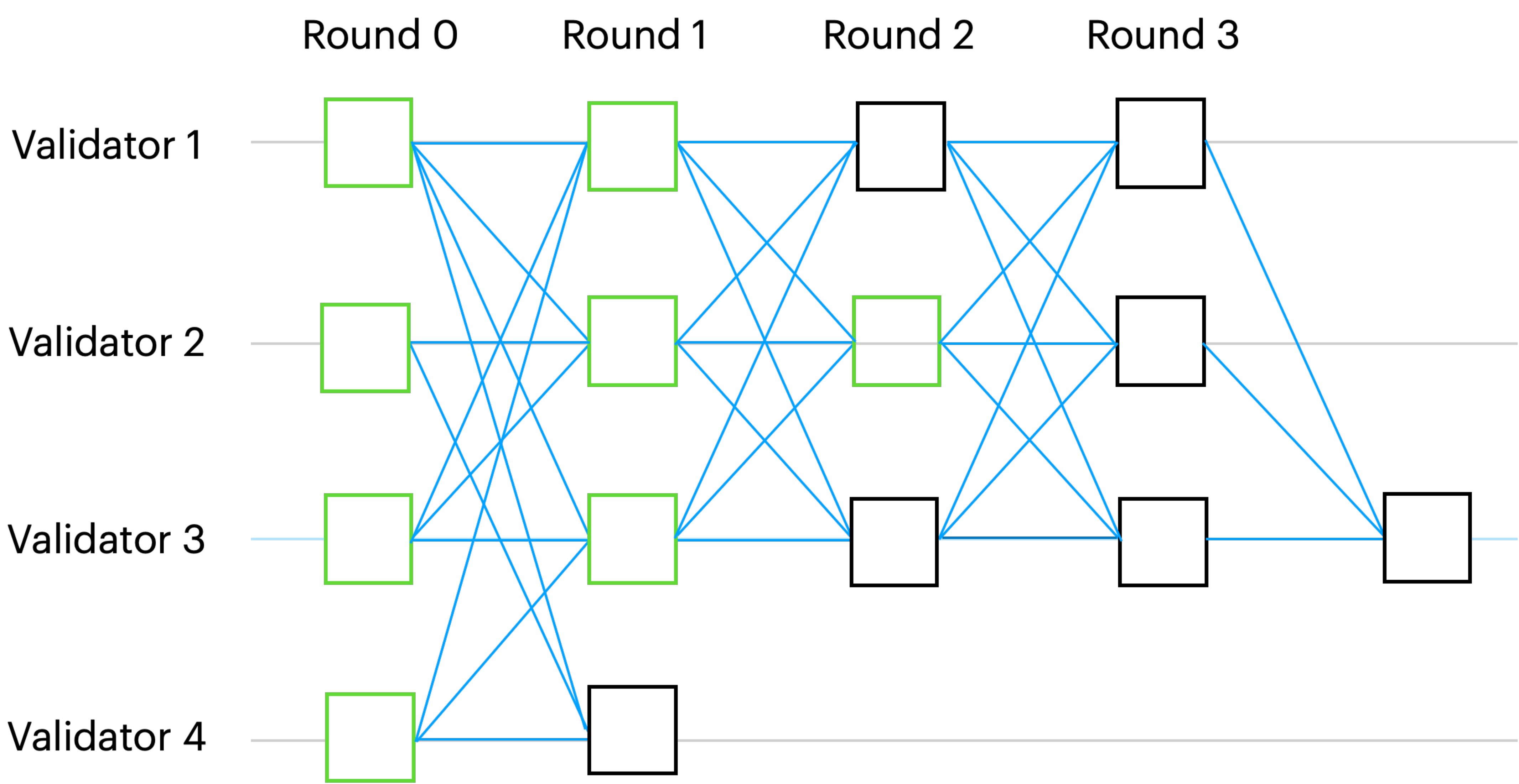}
    \caption{A possible local view of a round-based DAG. The causal history of the vertex identified by validator 2 in round 2 is highlighted in green.}
    
    \label{fig:dag}
\end{figure}

To disseminate messages, Narwhal uses an efficient reliable broadcast implementation that guarantees:
\begin{description}

\item \textbf{Validity:} if an honest validator has a vertex $\mathsf{v}$ in its local view of the DAG, then it also has all the causal history of $\mathsf{v}$.

\item \textbf{Eventual delivery:} if an honest validator has a vertex in round $\mathsf{r}$ by validator $\mathsf{p}$ in its local view of the DAG, then eventually all honest validators have a vertex in round $\mathsf{r}$ by validator $\mathsf{p}$ in their local views of the DAG. 

\item \textbf{Non-equivocation:} if two honest validators have a vertex in round $\mathsf{r}$ by validator $\mathsf{p}$ in their local views of the DAG, then the vertices are identical.
\end{description}

\noindent Inductively applying Validity and Non-equivocation, we get:

\begin{description}

\item \textbf{Completeness:} if two honest validators have a vertex $\mathsf{v}$ in round $\mathsf{r}$ by validator $\mathsf{p}$ in their local views of the DAG, then $\mathsf{v}$'s causal histories are identical in both validators' local view of the DAG.

\end{description}

In simple words, Narwhal construction guarantees that
\begin{enumerate}
    \item All validators eventually see the same DAG; and
    \item Any two validators that have the same vertex $v$ locally also agree on the whole causal history of $v$ (the contents of vertices and edges between them).
\end{enumerate}

\paragraph{DAG-Rider / Tusk / Bullshark}
DAG-Rider, Tusk, and Bullshark are all algorithms to agree on the total order of all vertices in the DAG with no additional communication overhead.
Each validator independently looks at its local view of the DAG and orders the vertices without sending a single message. This is done by interpreting the structure of the DAG as a consensus protocol, where a vertex represents a proposal and an edge represents a vote.

DAG-Rider~\cite{allyouneed} and Tusk~\cite{narwhaltusk} are randomized protocols designed to tolerate full asynchrony, which necessitates a larger number of rounds and consequently, a higher latency. Bullshark~\cite{bullsharksync} also provides a deterministic protocol variant with a faster ordering rule, relying on partial synchrony for liveness.
%Perhaps not surprisingly, all three protocols use quorum intersections in the DAG structure to guarantee safety.
While the specific details are not required to understand this paper, next we explain the high-level structure of these protocols and define a property they all share. 

\subsection{Common framework}
\label{sub:framework}
Narwhal-based consensus protocols have the following common abstract structure:
\begin{enumerate}
    \item Pre-determined anchors. Every few rounds (the number depends on the protocol) there is a round with a pre-determined leader. The vertex of the leader is called an \emph{anchor}. In the partially synchronous version of Bullshark, the leaders are a-priori known. In the asynchronous protocols (DAG-Rider, Tusk, asynchronous Bullshark) the leaders are hidden and revealed during the DAG construction.
    \item Order the anchors. All validators independently decide which anchors to skip and which to order. The details differ among the protocols, although they all rely on quorum intersection in the DAG structure. The key aspect is that each honest validator locally decides on a list of anchors, and all lists share the same prefix.
    \item Order causal histories. Validators process their list of ordered anchors one by one, and for each anchor order all previously unordered vertices in their causal history by some deterministic rule. By Completeness, all validators see the same causal history for any anchor, so all validators agree on the total order. 
\end{enumerate}

An illustration of the ordering logic appears in Figure~\ref{fig:ordering}.

\begin{figure}
    \centering
    \includegraphics[width=0.48\textwidth]{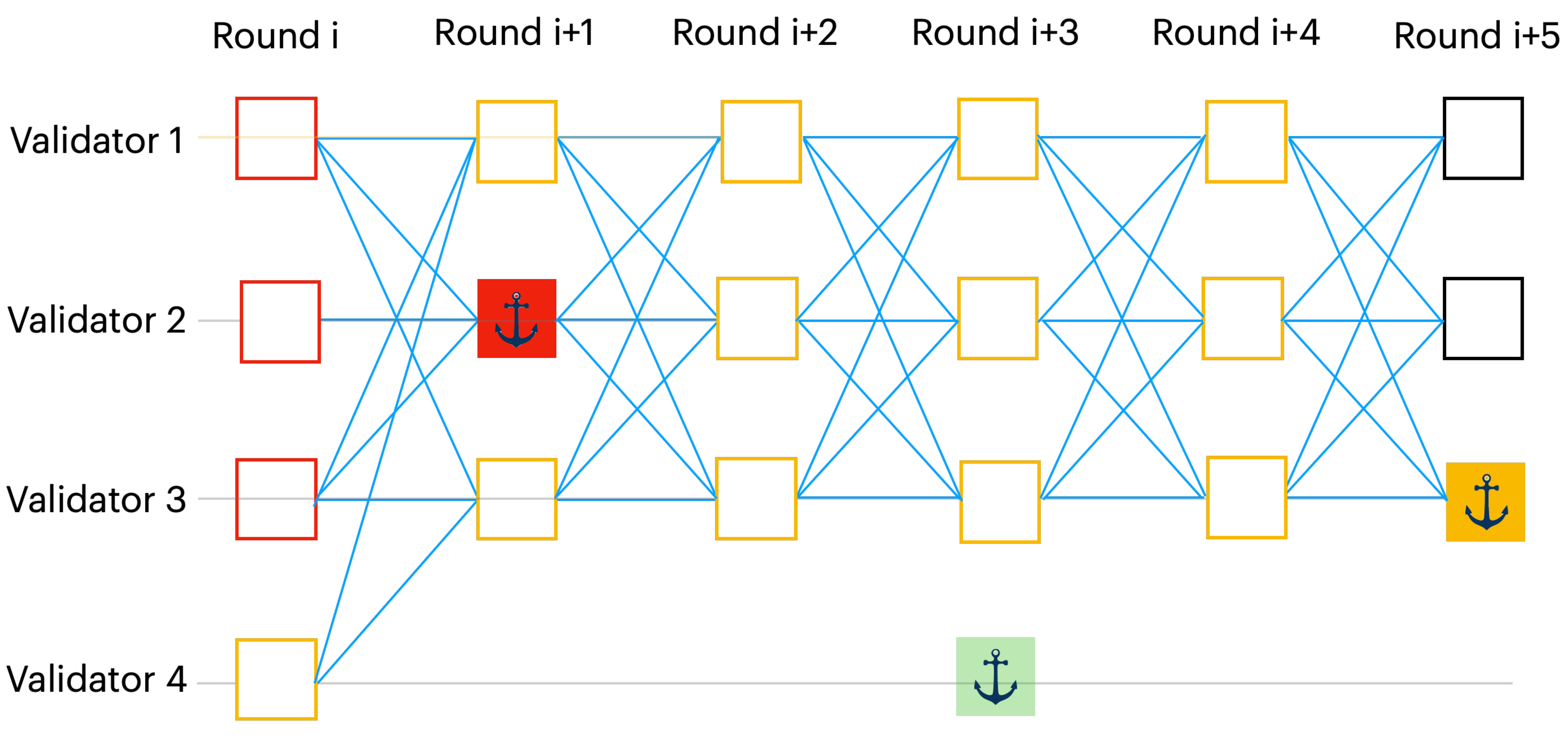}
    \caption{A possible local view of the DAG in the partially synchronous Bullshark protocol. Filled squares represent the pre-defined anchors. In this example, the validator orders the red and yellow anchors, while the green (which is not in the DAG) anchor is skipped. To order the DAG, the validator deterministically orders the red anchor's causal history (the unfilled red vertices) and immediately after the yellow anchor's causal history (the unfilled yellow vertices). 
    %\sasha{change to anchor image? Yes.}
    }
    
    \label{fig:ordering}
\end{figure}

%As mentioned above and illustrated in Figure~\ref{fig:ordering}, 
The key correctness argument for all the above mention consensus protocols relies on the fact that all validators agree on which anchors to order and which to skip. In particular, they will all agree on the first anchor that no validator skips. More formally, the abstract property of the Narwhal-based consensus protocols that our \sysname framework relies on is the following:

\begin{property}
\label{pro:firstanchor}
    Given a Narwhal-based protocol $\mathcal{P}$, if all honest validators agree on the mapping from rounds to leaders before the beginning of an instance of $\mathcal{P}$, then they will agree on the first anchor each of them orders during the execution of $\mathcal{P}$.
\end{property}

The proof follows immediately from Proposition 2 in DAG-Rider~\cite{allyouneed} and Corollary C. in Bullshark~\cite{bullshark}.

%A good validator in synchrnoy is never skipped (maybe move to analysis)

%A crash validator is always skipped. (for latency analysis)

\section{\sysname}

\sysname is protocol agnostic and can be directly applied to all Narwhal-based consensus protocols, i.e., DAG-Rider, Tusk, and Bullshark.
It makes no changes to the protocols but rather combines their instances in essentially a ``black-box" manner. 
The entire correctness argument can be derived solely from Property~\ref{pro:firstanchor}.

\subsection{Pipelining} 
\label{sec:pipelining}
A natural progression after the high throughput scalability of BFT consensus achieved by Narwhal is to reduce latency as much as possible.
To this end, Bullshark already halved DAG-rider's latency for ordering anchors from 4 rounds to 2 by adding an optimistic path under the partially synchronous network communication assumption.

Intuitively, it is hard to imagine latency lower than 2 rounds as in the interpretation of the DAG structure as a consensus protocol, one round is needed to "propose" the anchor, while another is needed for "voting". 
However, only anchors can be ordered in 2 rounds. 
The rest of the vertices are ordered as part of the causal history of some anchor and require a minimum latency of 3 or 4 rounds.
This is because the vertices in a "voting" round require (minimum) 3 rounds, while vertices that share a round with an anchor have to wait for at least the next anchor to be ordered, thus requiring (minimum) 4 rounds.
An illustration of the ordering latency for different vertices appears in Figure~\ref{fig:latency}.

\begin{figure}
    \centering
    \includegraphics[width=0.48\textwidth]{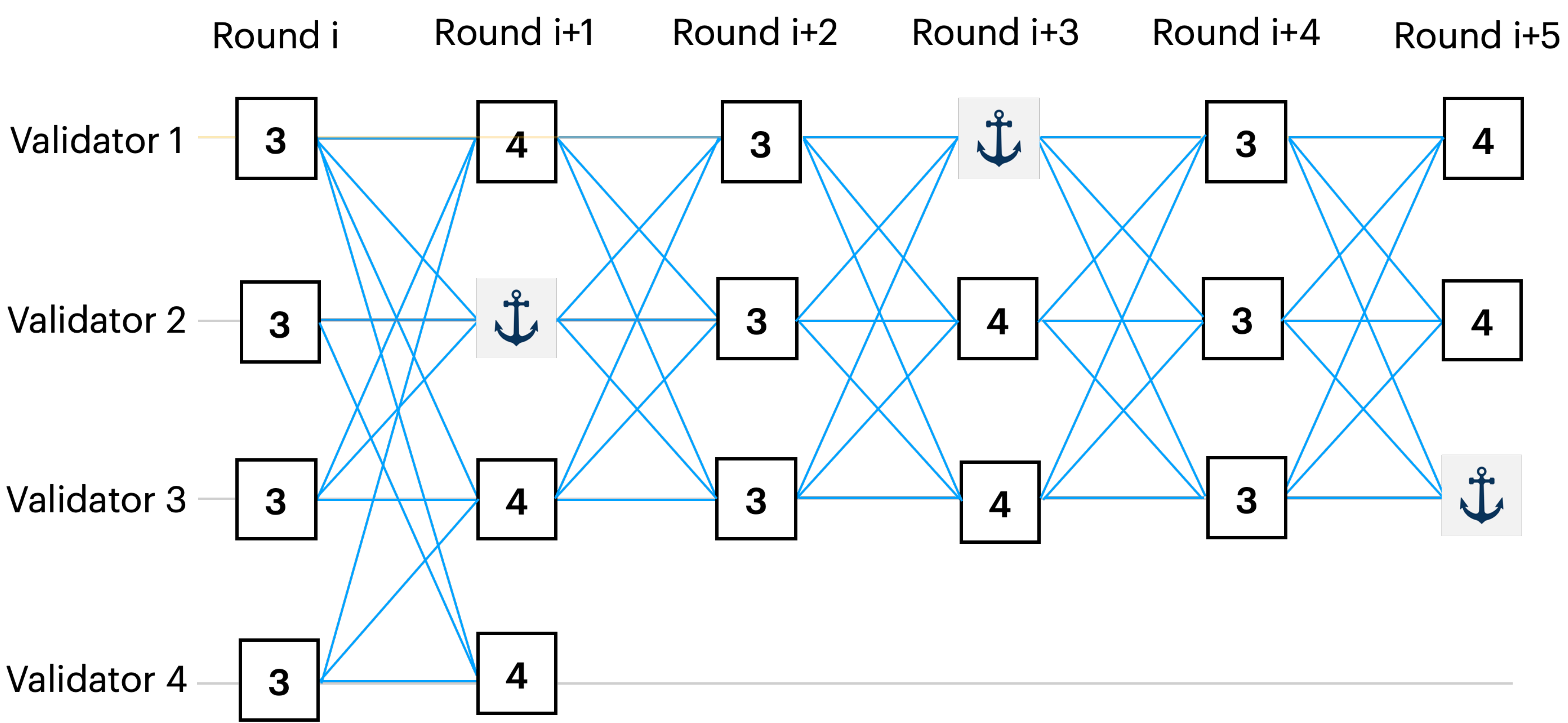}
    \caption{Illustration of the number of rounds required for each vertex in the DAG to be ordered in the best case, according to the Bullshark protocol. The number in each vertex represents its minimum latency. 
    For example, the anchor of round $i+1$ can be ordered in round $i+2$, but the other vertices in this round require at least 4 rounds to be ordered.}
    
    \label{fig:latency}
\end{figure}

Ideally, to reduce the latency of ordering vertices we would like to have an anchor in every round.
This would allow for non-anchor vertices to be ordered as a part of some anchor's causal history in each and every round, making latency and throughput of the protocol less spiky. 
%This would reduce the latency of all vertices to the minimum possible by the protocol's order rule and make the throughput less spiky. 
In Bullshark, it would become possible for every non-anchor vertex to be ordered in 3 rounds (see Figure~\ref{fig:latency}), while in DAG-Rider the latency may be reduced from 10 rounds to 7 in expectation. 

\paragraph{Solution}
Let $\mathcal{P}$ be any Narwhal-based consensus protocol.
On a high level, the core technique in \sysname is to execute $\mathcal{P}$ until it, as a consensus protocol, guarantees agreement on some part of the DAG for all validators.
%\sasha{This is not clear. I think we should link the property and talk directly about the first ordered anchor. Otherwise, what prefix?
%Also, below you say "the next round". Next after what?
%edit: moe details are given below so this flaf might be fine. But in any case I would start with "intuitively" or "on a high level" to emphasize that we are being imprecise here.}
Starting from the round following the agreed part of the DAG, all validators can switch over and start executing a new instance of $\mathcal{P}$ (or a different Narwhal-based consensus protocol, if desired) from scratch.
While the instances are not executing concurrently, this scheme effectively pipelines the ``proposing" and ``voting" rounds. As a result in \sysname, in a good case an anchor is ordered in every round.

The pseudocode appears in Algorithm~\ref{alg:pipelining}.
In the beginning of the protocol, all validators interpret the DAG from round $0$, and the function $F$ is some pre-defined deterministic mapping from rounds to leaders.
Each validator locally runs $\mathcal{P}$, using $F$ to determine the anchors, until it orders the first anchor, denoted by $A$ in round $r$.
The key is that, by the correctness of $\mathcal{P}$ as stated in Property~\ref{pro:firstanchor}, all validators agree that $A$ is the first ordered anchor (previous anchors are skipped by all validators).
Consequently, each validator can re-interpret the DAG from the next round (round $r+1$) according to a new instance of the protocol $\mathcal{P}$ (or another Narwhal-based protocol) executing from scratch from round $r+1$.

\begin{algorithm}
\caption{Pipelining in \sysname}
\begin{algorithmic}[1]
        \State $\emph{current\_round} \gets 0$
        \State $F: R \rightarrow A$
        \Comment{deterministic rounds to anchors mapping}
        \While{\emph{true}}
        \State execute $\mathcal{P}$, select anchors by $F$, starting from \StateX\emph{current\_round} until the first ordered (not skipped)
        \StateX anchor is determined.

        \State let $A$ be the first ordered anchor in round $r$
        \State order $A$'s causal history according to $\mathcal{P}$ 
        %\Statex \hspace{4.5mm} deterministic rule
        \State $\emph{current\_round} \gets r+1$ 
        %\State update $F$ according to $A$'s causal history
        %\label{line:updatingF2}
        %\State $F(x) \gets F(x - \emph{current\_round})$
        \EndWhile
\end{algorithmic}
\label{alg:pipelining} 
\end{algorithm}

To order the DAG, much like in the original $\mathcal{P}$, the validators deterministically order $A$'s causal history, and by the Completeness property, arrive at the same total order over the same vertices.
Note that without re-interpreting the DAG according to a new instance of $\mathcal{P}$ starting from round $r+1$, the next anchor according to the previously executing instance of the protocol would appear in a strictly later round (e.g. $r+4$ for DagRider and $r+2$ for Bullshark).
The above process can continue for as long as needed.
An illustration appears in Figure~\ref{fig:pipelining}.

\begin{figure}
    \centering
    \includegraphics[width=0.48\textwidth]{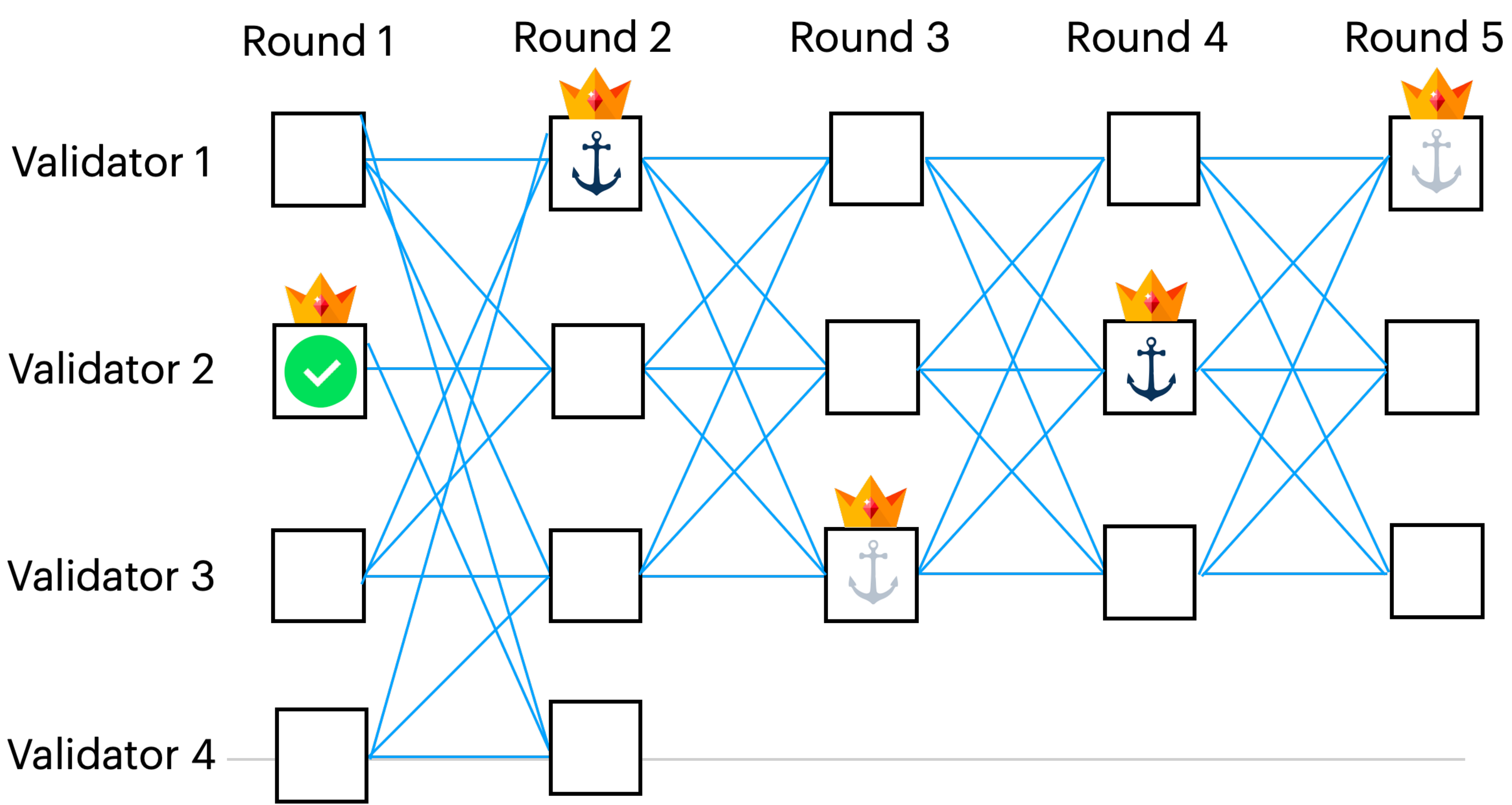}
    \caption{Illustration of \sysname's pipelining integrated into Bullshark.
    The vertices that are fixed to be anchors by $F$ are marked by a crown.
    The protocol starts by interpreting the DAG with anchors in rounds $1, 3$, and $5$. Bullshark determines that the anchor in round 1, marked by a green checkmark, is the first to be ordered. 
    Then, a new instance of Bullshark starts at round 2 with the anchors marked in rounds $2$ and $4$.}
    \label{fig:pipelining}
\end{figure}

Note that in Algorithm~\ref{alg:pipelining}, function $F$ is fixed and used by each instance of protocol $\mathcal{P}$. In a true "black-box" implementation, the round numbers could be different from the perspective of the executing protocol instance (i.e. start from $0$ for each new instance). However, $F$ is fixed and always assigns the same anchor to any given round $r$ in \sysname regardless of the protocol instance used for this round.

Note that with Shoal, ordering an anchor vertex requires 2 rounds, while all other vertices require 3. 
In Section~\ref{sec:multi} we discuss a potential direction to reduce the latency for non-anchor vertices by treating all vertices as anchors. 
Intuitively, we can use Property~\ref{pro:firstanchor} to instantiate a binary agreement to decide whether to commit each vertex individually.
%Leader reputation is a property that many times overlooked in BFT literature but has a significant impact on production performance.

\subsection{Leader Reputation}
\label{sec:leaderreputation}
BFT systems are designed to tolerate Byzantine failures in order to provide as strong as possible worst-case reliability guarantees.
However, actual Byzantine failures rarely occur in practice.
This is because validators are highly secured and have strong economic incentives to follow the protocol.
Slow or crashed leaders are a much more frequent occurrence which can significantly degrade the system performance.
In Narwhal-based BFT, if the leader of round $r$ crashes, no validator will have the anchor of round $r$ in its local view of the DAG. 
Thus, the anchor will be skipped and no vertices in the previous round can be ordered until some later point due to an anchor in a future round.

The way to deal with missing anchors is to somehow ensure that the corresponding leaders are less likely to be elected in the future.
%slow/crashed leaders that are not able to drive the protocol progress is to make sure they are less likely to be elected in the future.
A natural approach to this end is to maintain a reputation mechanism, assigning each validator a score based on the history of its recent activity. A validator that has been participating in the protocol and has been responsive would be assigned a high score. Otherwise, the validator is either crashed, slow, or malicious and a low score is assigned.
The idea is then to deterministically re-compute the pre-defined mapping from rounds to leaders every time the scores are updated, biasing towards leaders with higher scores. In order for validators to agree on the new mapping, they should agree on the scores, and thus on the history used to derive the scores.

Such a mechanism was previously proposed in~\cite{cohen2022aware} and implemented in the Diem Blockchain~\cite{diembft} to enhance the performance of Jolteon~\cite{jolteon}, a leader-based consensus protocol.
One important property Jolteon is that Safety is preserved even if validators disagree on the identity of the leader, while liveness is guaranteed as long as they eventually converge.
Hence, validators could re-assign the reputation scores every time a new block was committed, even though during asynchronous periods it was possible for different validators to commit the same block in different rounds.
Unfortunately, this is not the case for Narwhal-based BFT. If validators disagree on the anchor vertices, they will order the DAG differently and thus violate safety.
This makes the leader reputation problem strictly harder in Narwhal-based BFT. 

\paragraph{Solution}
\sysname constructs a protocol identical to a given Narwhal-based consensus protocol $\mathcal{P}$, but to support leader reputation anchors are selected according to a function $F$ that takes into account validators' recent activity, e.g., the number of vertices they have successfully added to the DAG.
The function $F$ should be updated as frequently as possible and aim to select validators with a better reputation as leaders more often than their counterparts with a lower reputation.

In \sysname, pipelining and leader reputation can be naturally combined as they both utilize the same core technique of re-interpreting the DAG after agreeing on the first ordered anchor. 
In fact, the pseudocode for \sysname appears in Algorithm~\ref{alg:leaderreputation} only differs from Algorithm~\ref{alg:pipelining} by adding line~\ref{line:updatingF2}.
The idea is that the validators simply need to compute a new mapping, starting from round $r+1$, based on the causal history of ordered anchor $A$ in round $r$ (which they are guaranteed to agree on by Property~\ref{pro:firstanchor}). Then, the validators start executing a new instance of $\mathcal{P}$ from round $r+1$ with the updated anchor selection function $F$. 

\begin{algorithm}
\caption{\sysname}
\begin{algorithmic}[1]
        \State $\emph{current\_round} \gets 0$
        \State $F: R \rightarrow A$
        \Comment{deterministic rounds to anchors mapping}
        \While{\emph{true}}
        \State execute $\mathcal{P}$, select anchors by $F$, starting from \StateX\emph{current\_round} until the first ordered (not skipped)
        \StateX anchor is determined.

        \State let $A$ be the first ordered anchor in round $r$
        \State order $A$'s causal history according to $\mathcal{P}$ 
        %\Statex \hspace{4.5mm} deterministic rule
        \State $\emph{current\_round} \gets r+1$ 
        \State update $F$ according to $A$'s causal history
        \label{line:updatingF2}
        \EndWhile
\end{algorithmic}
\label{alg:leaderreputation} 
\end{algorithm}

Our solution is protocol agnostic and can be directly applied to all Narwahl-based consensus protocols, i.e., DAG-Rider, Tusk, and Bullshark. An illustration can be found in Figure~\ref{fig:leaderreputation}.
\sysname makes no changes to the protocols but rather combines their instances, and the entire correctness argument can be derived solely from Property~\ref{pro:firstanchor}.

%The illustration can be found in Figure~\ref{fig:leaderreputation}

%The pseudocode appears in Algorithm~\ref{alg:leaderreputation}.
%At the beginning of the protocol, all validators interpret the DAG from round $0$, and the function $F$ is some pre-defined deterministic mapping. Each validator locally runs $\mathcal{P}$, using $F$ to determine the anchors, until it orders the first anchor, denoted by $A$ in round $r$.
%By the correctness of $\mathcal{P}$, as stated in Property~\ref{pro:firstanchor}, all validators agree that $A$ is the first ordered anchor.
%Consequently, each validator can re-interpret the DAG from the next round (round $r+1$) according to a new instance of the protocol $\mathcal{P}$ (or another Narwhal-based protocol) executing from scratch from round $r+1$.

%To incorporate leader reputation, the validators update the function $F$ according to the information in $A$'s causal history. 
%By completeness, all validators observe the same causal history of anchor $A$, and thus they all make identical updates to the function $F$.
%Thus, validators order $A$'s causal history and start executing a new instance of $\mathcal{P}$ from round $r+1$ with the updated anchor selection function $F$. 

\begin{figure}
    \centering
    \includegraphics[width=0.48\textwidth]{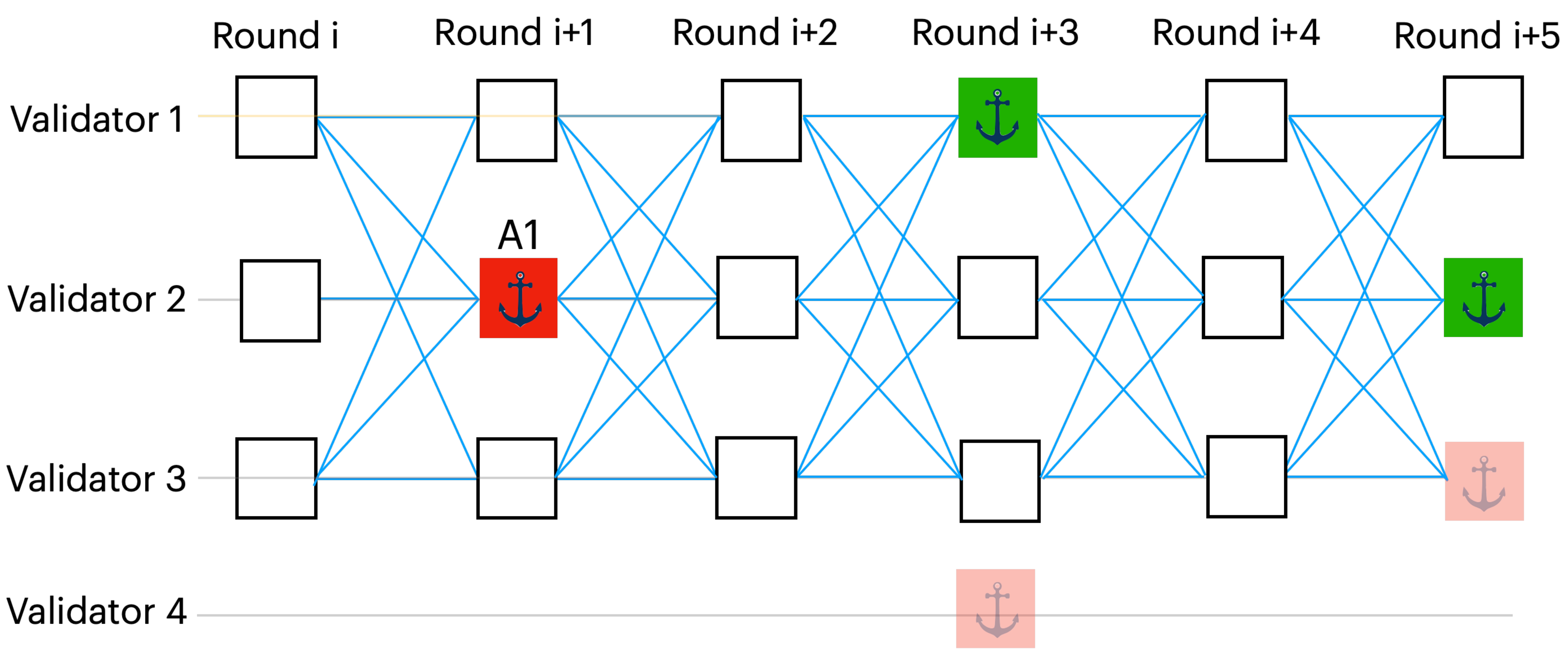}
    \caption{Illustration of \sysname's leader reputation integrated into Bullshark (no pipelining). 
    First, the DAG is interpreted via the Bullshark protocol and the red anchors. The anchor in round i+1, A1, is determined to be the first ordered anchor. 
    Then, based on A1's causal history, new anchors are selected for future rounds (marked in green). 
    Note that validator 4, which had an anchor according to the red selection, no longer has an anchor according to the new mapping (it was not performing well).
    Then, A1's causal history is deterministically ordered as in the original Bullshark, and a new instance of Bullshark starts at round $i+2$ based on the green anchors.}
    
    \label{fig:leaderreputation}
\end{figure}

%Correctness of the scheme still follows from Property~\ref{pro:firstanchor}.

%Each validator starts executing an instance of $\mathcal{P}$ until it determines the first non-skipped anchor (in other words, the anchor that is ordered first) according to the local deterministic ordering rule of $\mathcal{P}$ based on the validator's local view of the DAG. 
%The key is that by Property~\ref{pro:firstanchor} all validators agree on which anchor is ordered first.
%Let $A$ be this anchor in round $r$.

%Since all validators agree on the same anchor $A$, they can all start a new instance of $\mathcal{P}$ from round $r+1$. 

%In order to start a new instance of $\mathcal{P}$ in round $r+1$, validators need to agree on the mapping from rounds to leaders.

%One way is to use the same mapping as the original one but shift it to start from round $r+1$.

%As we analyze below, in the best case \sysname orders an anchor every round.  
%For Bullshark, this means anchors have 2 rounds of latency and the rest 3.

%\input{sections/further}
%\section{Analysis}\label{sec:proofs}
%\sysname is a framework to add leaser reputation and pipelining to any Narwhal-based DAG-BFT protocol.
%As such, its security and latency depend on the underlying protocol. 

%\subsection{Security}\label{sub:security}
\section{Correctness}\label{sec:security}
To prove the correctness of \sysname (Algorithm~\ref{alg:leaderreputation}) we assume that the underlying protocol satisfies Property~\ref{pro:firstanchor}, which we will use inductively. 
%The proofs for leader reputation and pipelining are very similar and thus we provide below one proof that captures both.

\begin{lemma}
\label{lem:safety}
    Let $P$ be a Narwhal-based DAG-BFT protocol that satisfies Property~\ref{pro:firstanchor}.
    Let $D$ be a round-based DAG, and assume a known to all function $F$ that maps rounds to anchors.
    Then all the locally ordered lists of anchors by honest validators executing \sysname with $P$ according to $F$ share the same prefix.
\end{lemma}

\begin{proof}
    Proof is by induction on the ordered anchors.

    \textbf{Base:} We need to show that all honest validators agree on the first anchor. 
    Since \sysname starts by running $P$ until the first anchor is ordered, the base case follows immediately from Property~\ref{pro:firstanchor}.

    \textbf{Step:} Assume all honest validators agree on the first $k$ ordered anchors, we need to prove that they agree on anchor $k+1$.
    First, we show that all honest validators agree on the new function $F$ (Line~\ref{line:updatingF2} in Algorithm~\ref{alg:leaderreputation}). 
This holds because the new function $F$ is deterministically computed according to the information in $k$'s causal history, and by the Completeness property of the DAG, all honest validators have the same causal history of anchor $k$ in their local view.

    Next, let $r$ be the round of anchor $k$.
    By the inductive assumption, all honest validators agree on $r$.
    Thus, all honest validators start the next instance of $P$ in the same round $r+1$.

    Now consider a DAG $D'$ that is identical to $D$ except it does not have the first $r$ rounds.
    By Property~\ref{pro:firstanchor}, all validators that run $P$ with the new function $F$ on $D'$ agree on the first ordered anchor in $D'$.
    Therefore, all validators agree on anchor $k+1$ in $D$.
\end{proof}

\begin{theorem}
    Let $P$ be a Narwhal-based DAG-BFT protocol that satisfies Property~\ref{pro:firstanchor}.
    \sysname with $P$ satisfies total order.
\end{theorem}

\begin{proof}
    By Lemma~\ref{lem:safety}, all validators order the same anchors. The theorem follows from the DAG Completeness property as all validators follow the same deterministic rule to order the respective causal histories of the ordered anchors.  
\end{proof}

\section{Implementation and Prevalent Responsiveness} \label{sec:implementation}

We have implemented Narwhal and the partially synchronous version of Bullshark as part of a publicly available open-source blockchain project\footnote{In order to uphold the anonymity requirement of the submission, we do not disclose the name of the blockchain project.}. This blockchain is live and the process of productionizing our implementation is underway. 
The code is written in Rust, utilizing Tokio\footnote{\url{https://tokio.rs}} for asynchronous networking, BLS~\cite{bls2001} implemented over BLS12-381 curves for signatures, RocksDB\footnote{\url{https://rocksdb.org}} for persistent data storage, and the Noise\footnote{\url{https://github.com/noiseprotocol/noise_spec}} protocol for authenticated messages.

%Our implementation is fully integrated into the system, meaning that our evaluation relies on the system's networking stack. 
%As mentioned in the introduction, we improve Narwhal and Bullshark with a few system-level optimizations that improve latency.
%In our evaluation, we compare three versions, which are described below.

\subsection{Vanilla Bullshark}
We implemented Bullshark according to~\cite{bullshark}, but additionally incorporated weak links per~\cite{allyouneed} in our DAG construction.
Observing $n-f$ vertices in a round is sufficient for progressing to the next round.
Therefore, without weak links, slow validators may consistently lag behind others in broadcasting their vertices and thus may consistently fail to add their vertices to the DAG.  
This will incur significant latency for their client transactions.
Weak links from a vertex can reference vertices from earlier rounds in addition to the normal (strong) links to $n-f$ vertices from the previous round. 
These weak links are used when establishing the causal history of ordered anchors and thus facilitate the inclusion of transactions contributed by the slow validators into the total order.

%These links do not influence the Bullshark ordering logic of anchors but are utilized to establish the causal history of anchors. Consequently, slow validators can contribute vertices to the DAG, ensuring the ordering of their transactions. 
We refer to this implementation as Vanilla Bullshark.
It is important to note that adding the support for weak links increases the average latency compared to the figures presented in~\cite{bullshark}, which did not employ the weak links. 

\subsection{Eliminating Timeouts}
The short paper for the stand-alone partially synchronous version of Bullshark~\cite{bullsharksync} assumes the DAG is given and focuses on the ordering of its vertices. On the other hand, full Bullshark is an asynchronous protocol with a fast path under partial synchrony. The full Bullshark paper~\cite{bullshark} describes how to build the DAG and in particular, the incorporation of timeouts to support the fast path.

Validators in Bullshark must observe $n-f$ vertices in a round to advance to the next round.
Even rounds have anchors, while vertices in odd rounds determine the ``voting" pattern.
Full Bullshark uses the following timeouts for every validator to support the fast path:
\begin{itemize}
    \item Even-round: wait until the anchor of the round is delivered (or the timeout expires).
    \item Odd-round: wait until $2f+1$ vertices that link to the anchor in the previous round are delivered (or the timeout expires). 
\end{itemize}
The rationale for the above logic is to help order the anchor within 2 rounds.
However, part of the contribution of this paper is to eliminate these timeouts in such a way that actually significantly improves latency, according to our evaluation. Having fewer cases where timeouts can occur also inherently simplifies the potential state space and thus, the implementation of the protocol. 
In Section~\ref{sec:evaluation}, we refer to even-rounds as \emph{anchor rounds} and to odd-rounds as \emph{vote rounds.}

\paragraph{Vanilla bullshark w/o vote Timeout}
In the full Bullshark $2f+1$ votes are required to order anchors.
Without timeouts in odd rounds, a Byzantine adversary can prevent the fast pass from making progress even during synchrony.
As long as Byzantine validators deliberately not link to the anchor, and even $1$ of their vertices get delivered among the first $2f+1$ to an honest validator in an odd round, then the honest validator will not be able to order the anchor.

%If vertices originated by Byzantine validators deliberately do not link (vote) to the anchor, then it is enough for a validator to deliver $1$ of them out of the first $2f+1$ to advance round without ordering.
%The timeout, therefore, makes sure that validators wait for all vertices by honest validators, and thus order the anchors before advancing rounds.
%The full Bullshark needs the odd-round timeout to deal with Byzantine failures.
%Basically, since $2f+1$ votes are required to order anchors, without the timeouts, an adversary can prevent the fast pass progress even during synchrony. 

However, we discovered that we can completely eliminate timeouts in odd rounds in the partially synchronous variant of Bullshark.
%However, we while implementing the stand-alone partially synchronous version, we realized that the odd-round timeout is redundant. 
The anchor ordering rule in this case is $f+1$ votes~\cite{bullsharksync}. 
As a result, even if $f$ out of the first $2f+1$ vertices delivered to a validator in a round is from Byzantine validators (and do not link to the anchor), the remaining $f+1$ vertices will link to the anchor due to the even-round timeout and be sufficient to order it.

%The baseline Bullshark is identical to vanilla Bullshark but without the 
\paragraph{Baseline Bullshark}
FLP impossibility result~\cite{flp} dictates that any deterministic protocol providing liveness under partial synchrony must use timeouts.
In Bullshark, without timeouts in the even rounds, an honest leader that is even slightly slower than the fastest $2f+1$ validators will struggle to get its anchor linked by other vertices.
As a result, the anchor is unlikely to be ordered.
%For Bullshark, without the even-round timeout, if a validator corresponding to an anchor is honest but slightly slower than the fastest $2f+1$ validators, then no validator will link to the anchor, and the anchor will not be ordered.
The timeout, therefore, ensures that all honest validators link to anchors during periods of synchrony (as long as the leader has not crashed and actually broadcasts the anchor vertex).

Even though timeouts are unavoidable in the worst case, we observe that the DAG construction combined with the leader reputation mechanism allows avoiding them in vast majority of cases in practice.
This is in contrast to leader-based monolithic consensus protocols, where timeouts are the only tool to bypass the rounds with bad leaders.
Without timeouts, a monolithic protocol could stall forever as there is no other mechanism to stop waiting for a crashed leader.
It is also hard to set the timeouts appropriately: conservative timeouts lead to excessive waiting for crashed leaders, while aggressive timeouts lead to bypassing slower validators (and hence unnecessarily failed rounds).

In contrast, the DAG construction provides a ``clock" that estimates the network speed.
Even without timeouts, the rounds keeps advancing as long as $2f+1$ honest validators continue to add their vertices to the DAG.
%Without timeouts, since $2f+1$ vertices are enough to advance rounds, validators will keep building the DAG with honest validators.
%Therefore, validators will keep trying to order new anchors indefinitely. 
As a result, the DAG can evolve despite some leaders being faulty. 
Eventually, when a non-faulty leader is fast enough to broadcast the anchor, the ordering will also make progress. 

Recall that to be ordered, in partially synchronous Bullshark, an anchor needs $f+1$ votes (links) out of the $3f+1$ vertices.
Therefore, as our evaluation demonstrates, in the failure-free case, most of the anchors are ordered in the next round.
The benefit are even more pronounced when there are failures. 
This is because a crashed validator causes a timeout to expire, stalling the protocol for the entire duration. 
Without a timer, however, the DAG will advance rounds at network speed and the Bullshark protocol is able to immediately move to the next anchor.

\paragraph{Timeouts as a fallback}
By FLP~\cite{flp} impossibility result, there exists an adversarial schedule of events that can prevent all anchors from getting enough votes to be ordered.
This scenario is extremely unlikely to occur in practice, but to be on the safe side, the protocol can deal with it by falling back to using timeouts after a certain amount of consecutive skipped anchors.

\subsection{\sysname of Bullsharks}
A realistic case in which timeouts can help the performance of a Narwhal-based consensus protocol is when the leader is slower than other validators.
Then, as discussed earlier, waiting for an anchor to be delivered even after $2f+1$ other vertices can allow the anchor to be committed in the next round.
While we eliminated timeouts from partially synchronous Bullshark, note that, due to the leader reputation mechanism, \sysname instantiated with Bullshark does better than repeatedly waiting for the slow leaders.
Instead, the leader reputation mechanism excludes (or at least significantly reduces the chances of) slow validators from being selected as leaders.
This way, the system takes advantage of the fast validators to operate at network speed.

\paragraph{Prevalent Responsiveness}
\sysname provides network speed responsiveness under all realistic failure and network scenarios, a property we name \emph{Prevalent Responsiveness}.
Specifically, compared to optimistic responsiveness, \sysname continues to operate at network speed even during asynchronous periods or if leaders fail for a configurable number of consecutive rounds. 

We implemented leader reputation and pipelining on top of the Baseline Bullshark and compared it to the baseline (no timeouts) implementation.    

\paragraph{Leader reputation logic}
As explained in Section~\ref{sec:leaderreputation}, \sysname ensures all validators agree on the information used to evaluate the recent activity and to bias the leader selection process accordingly towards healthier validators.
Any deterministic rule to determine the mapping from rounds to leaders (i.e. the logic in pseudocode Line~\ref{line:updatingF2} in Algorithm~\ref{alg:leaderreputation}) based on this shared and agreed upon information would satisfy the correctness requirements.
Next, we discuss the specific logic used in our implementation.

At any time each validator is assigned either a high or a low score, and all validators start with a high score.
After ordering an anchor $v$, each validator examines $v$'s causal history $H$.
Every skipped anchor in $H$ is (re-)assigned a low score, and every ordered anchor in $H$ is (re-)assigned a high score.
Then, the new sequence of anchors is pseudo-randomly chosen based on the scores, with a validator with a high score more likely to be a leader in any given round.
Note that while the validators use the same pseudo-randomness (so that they agree on the anchors), the computation is performed locally without extra communication. 

Assigning higher scores to validators whose anchors get ordered ensures that future anchors correspond to faster validators, thus increasing their probability to be ordered.
However, we ensure that the low score is non-zero, and thus underperforming validators also get a chance to be leaders. This crucially gives a temporarily crashed or underperforming validator a chance to recover its reputation.
\section{Evaluation} 
\label{sec:evaluation}

We evaluated the performance of the aforemententioned variants of Bullshark and \sysname on a geo-replicated environment in Google Cloud. 
In order to show the improvements from pipelining and leader reputation independently, we also evaluate \sysname PL, which is a \sysname instantiation with only pipelining enabled, and \sysname LR, which is a \sysname instantiation with only Leader Reputation enabled. 
With our evaluation, we aim to show that 
(i) \sysname maintains the same throughput guarantees as Bullshark.
(ii) \sysname can provide significantly lower latency than Bullshark and its variants.
(iii) \sysname is more robust to failures and can improve latency with the help of Leader Reputation.

For completeness, we also compare against Jolteon~\cite{jolteon}, which is the current consensus protocol of the production system we use.
Jolteon combines the linear fast path of Tendermint/Hotstuff with a PBFT style view-change, and as a result, reduces Hotstuff latency by 33\%. 
The implementation extends the original Jolteon protocol with a leader reputation mechanism, which prioritizes well-behaved leaders from previous rounds for future rounds.
In addition, to mitigate the leader bottleneck and support high throughput, the implementation uses the Narwhal technique to decouple data dissemination via a pre-step component (called Quorum Store~\cite{quorum-store}). 

\begin{figure}[t]
    \includegraphics[width=\columnwidth]{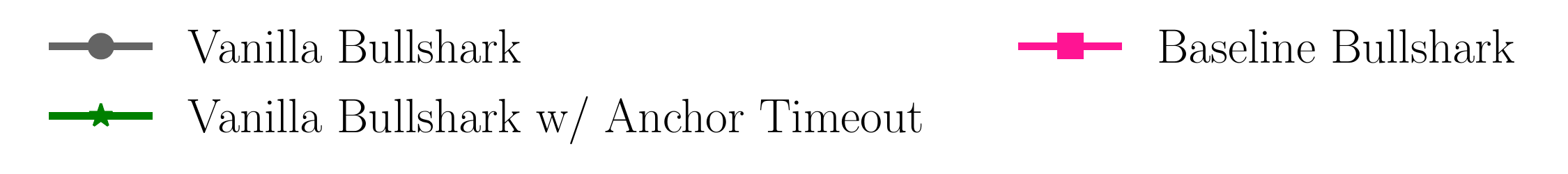}
    \begin{subfigure}[b]{0.49\columnwidth}
	\centering
	\caption{Throughput}
	\includegraphics[width=\textwidth]{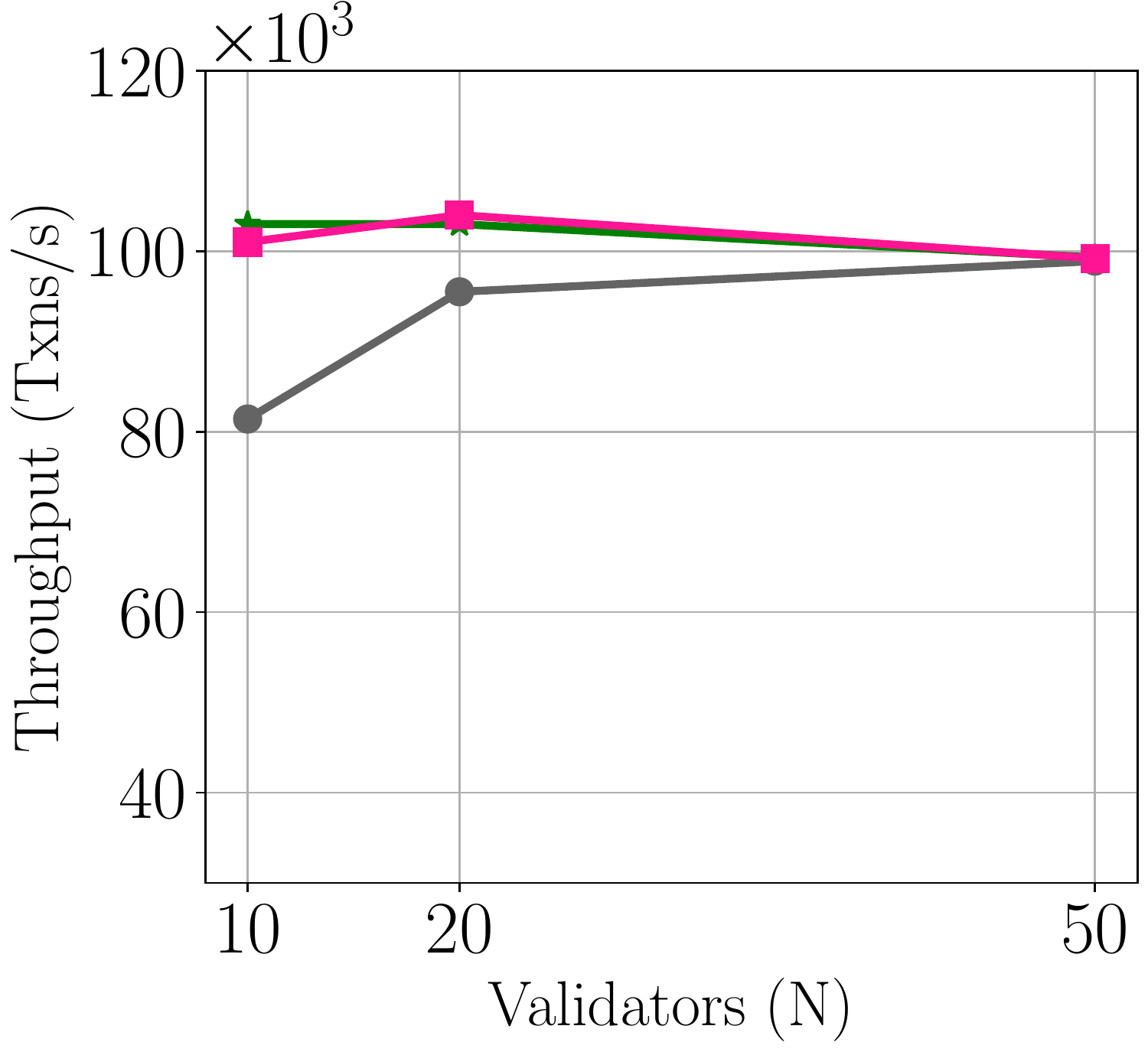}
	\label{fig:base-nf-tps}
    \end{subfigure}
    \begin{subfigure}[b]{0.49\columnwidth}
        \centering
	\caption{Average Latency}
        \includegraphics[width=\textwidth]{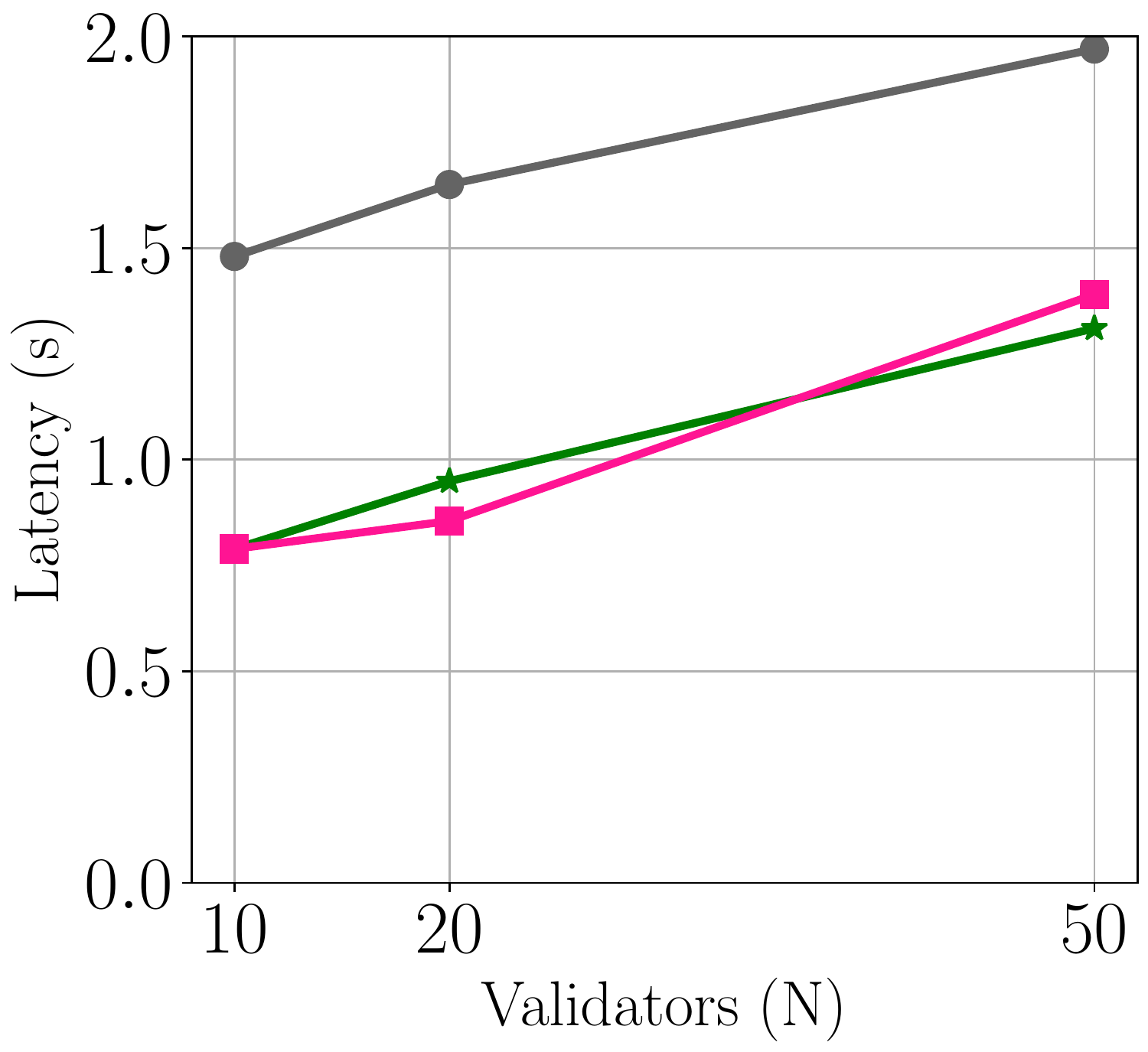}
	\label{fig:base-nf-lat}
    \end{subfigure}
    \vspace{-10pt}
    \caption{Baseline performance under no failures}
    \label{fig:base-nf-plots}
\end{figure}

We evaluate prevalent responsiveness by presenting experiments that compare variants of Bullshark w.o. timeout in different rounds versus \sysname as discussed in Section~\ref{sec:implementation}.

\textbf{Experimental Setup.} 
Our experimental setup consists of \texttt{t2d-standard-32} type virtual machines spread equally across three different Google Cloud regions: us-west1, europe-west4, asia-east1. 
Each virtual machine has 32 vCPUs, 128GB of memory, and can provide up to 10Gbps of network bandwidth. 
The round-trip latencies are: 
118ms between us-west1 and asia-east1, 
251ms between europe-west4 and asia-east1, 
and 133ms between us-west1 and europe-west4. 
The experiments involve three different values of N (the number of validators): 10, 20, and 50, tolerating up to 3, 6, and 16 failures, respectively.

We only measure the consensus performance to avoid introducing noise from other parts of the production system, such as execution and storage. 
The transactions are approximately 270B in size. We set a maximum batch size of 5000 transactions.
In our experiments, we measure \emph{latency} as the time elapsed from when a vertex is created from a batch of client transactions to when it is ordered by a validator. The timeouts for moving to the next round, when applicable, are set to 1s, which is less than the 1.5s timeout used by the production Blockchain system we use.

\begin{figure}[t]
    \begin{subfigure}[b]{0.49\columnwidth}
	\centering
	\caption{Throughput}
	\includegraphics[width=\textwidth]{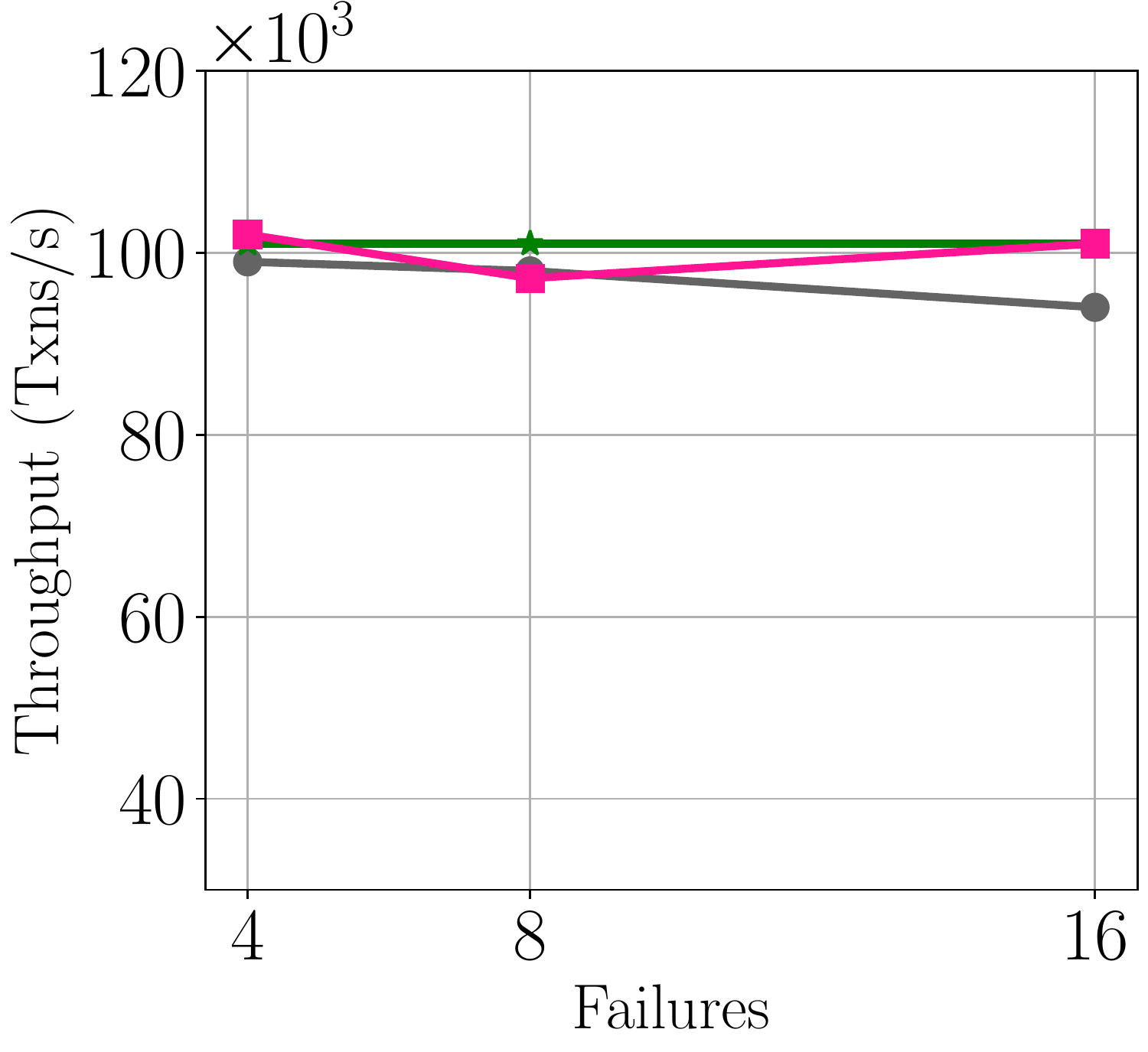}
	\label{fig:base-f-tps}
    \end{subfigure}
    \begin{subfigure}[b]{0.49\columnwidth}
        \centering
	\caption{Average Latency}
        \includegraphics[width=\textwidth]{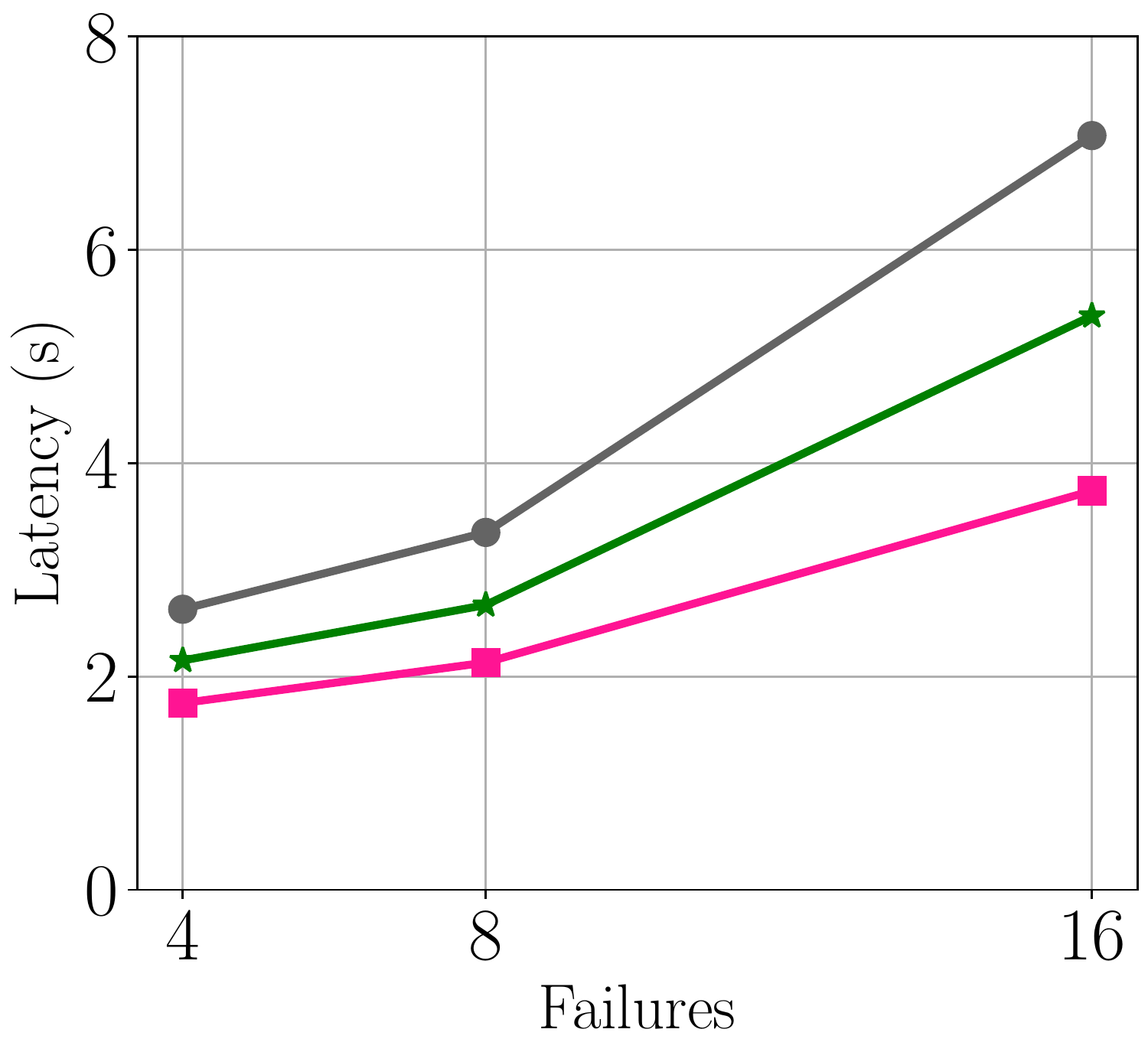}
	\label{fig:base-f-lat}
    \end{subfigure}
    \vspace{-10pt}
    \caption{Baseline performance under failures (N=50)}
    \label{fig:base-f-plots}
\end{figure}

\subsection{Baseline Performance}

First, we evaluate the performance of the Bullshark variants, namely Vanilla Bullshark, Vanilla Bullshark w/ Anchor Timeouts, and Baseline Bullshark, to align on a baseline performance to evaluate \sysname in the rest of the experiments. The results are in Figures~\ref{fig:base-nf-plots} and~\ref{fig:base-f-plots}. 

Figure~\ref{fig:base-nf-plots} shows the throughput and average latencies of the three Bullshark variants as the system size increases. The presence of timeouts in Vanilla Bullshark forces it to build the DAG slowly, which combined with the fact that fewer validators contribute vertices to the DAG when $N=10$, results in lower throughput than other variants, which have fewer or no timeouts. The latencies for Vanilla Bullshark is up to 88\% higher due to the timeouts. Interestingly, the latencies are similar for baseline Bullshark and Vanilla Bullshark w/o Vote timeout in the normal case because there is a trade-off between building a DAG at network-speed while skipping an anchor and waiting slightly longer for the anchor to be part of the votes. 

We also evaluated the vanilla variants and the baseline for $N=50$ and with varying the number of failures, in Figure~\ref{fig:base-f-plots}. We observe that Baseline Bullshark provides lower latency than other variants by virtue of being able to build the DAG at network speed skipping failed anchors and ordering using the alive ones. Therefore, in the rest of the section, we use Baseline Bullshark as the baseline to evaluate \sysname.

\begin{figure}[t]
    \includegraphics[width=\columnwidth]{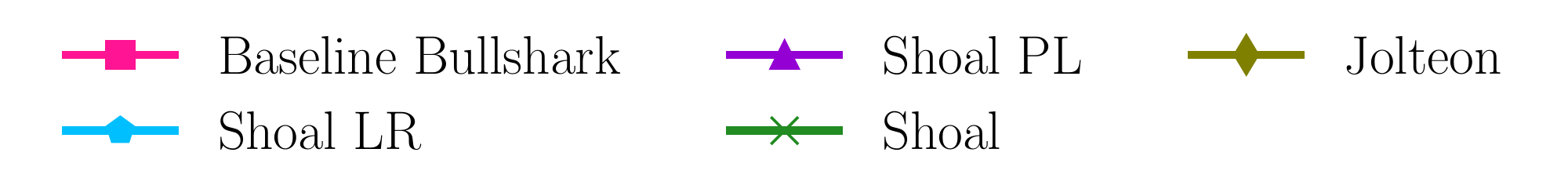}
    \begin{subfigure}[b]{0.49\columnwidth}
	\centering
	\caption{Throughput}
	\includegraphics[width=\textwidth]{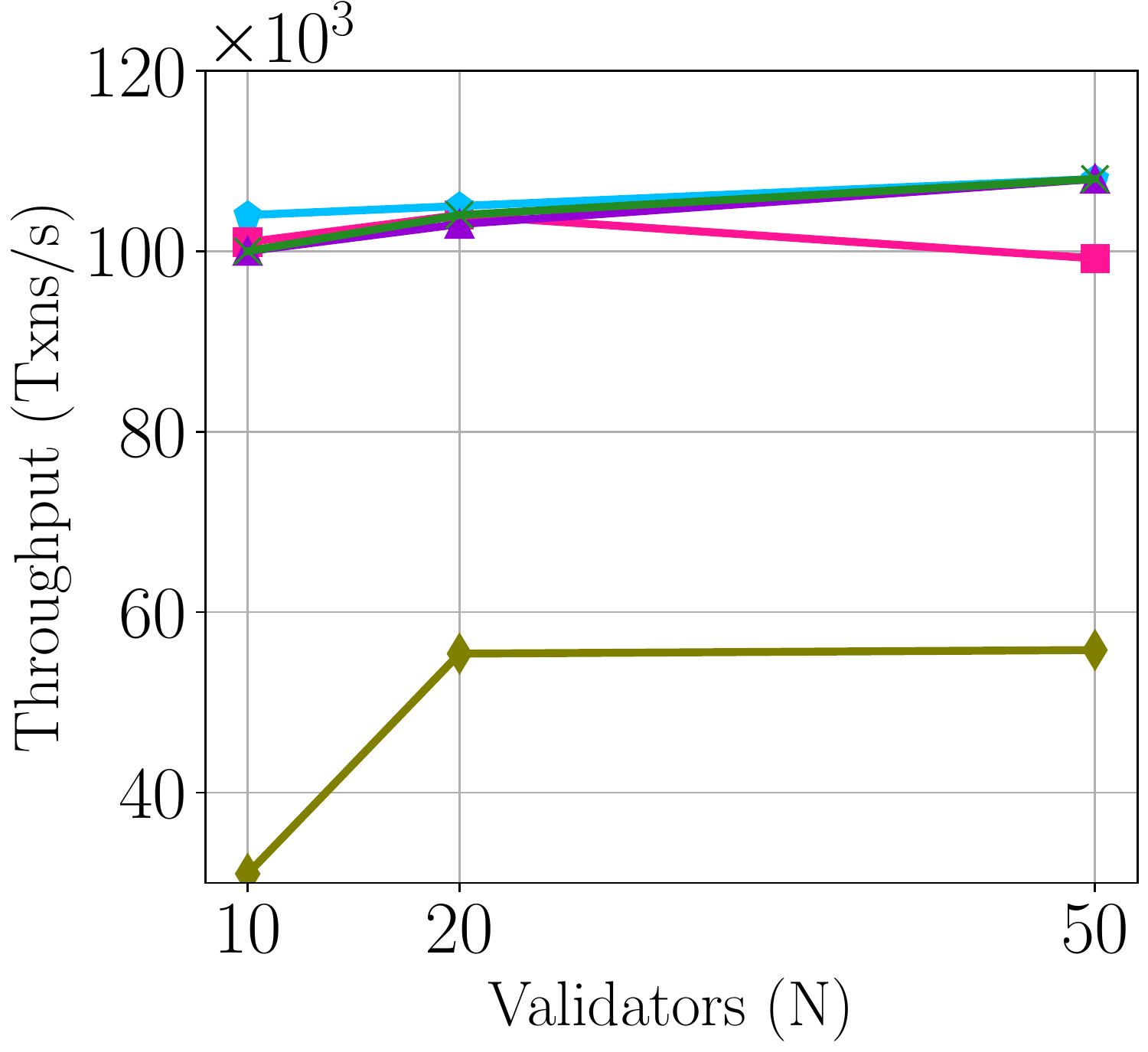}
	\label{fig:shoal-nf-tps}
    \end{subfigure}
    \begin{subfigure}[b]{0.49\columnwidth}
        \centering
	\caption{Average Latency}
        \includegraphics[width=\textwidth]{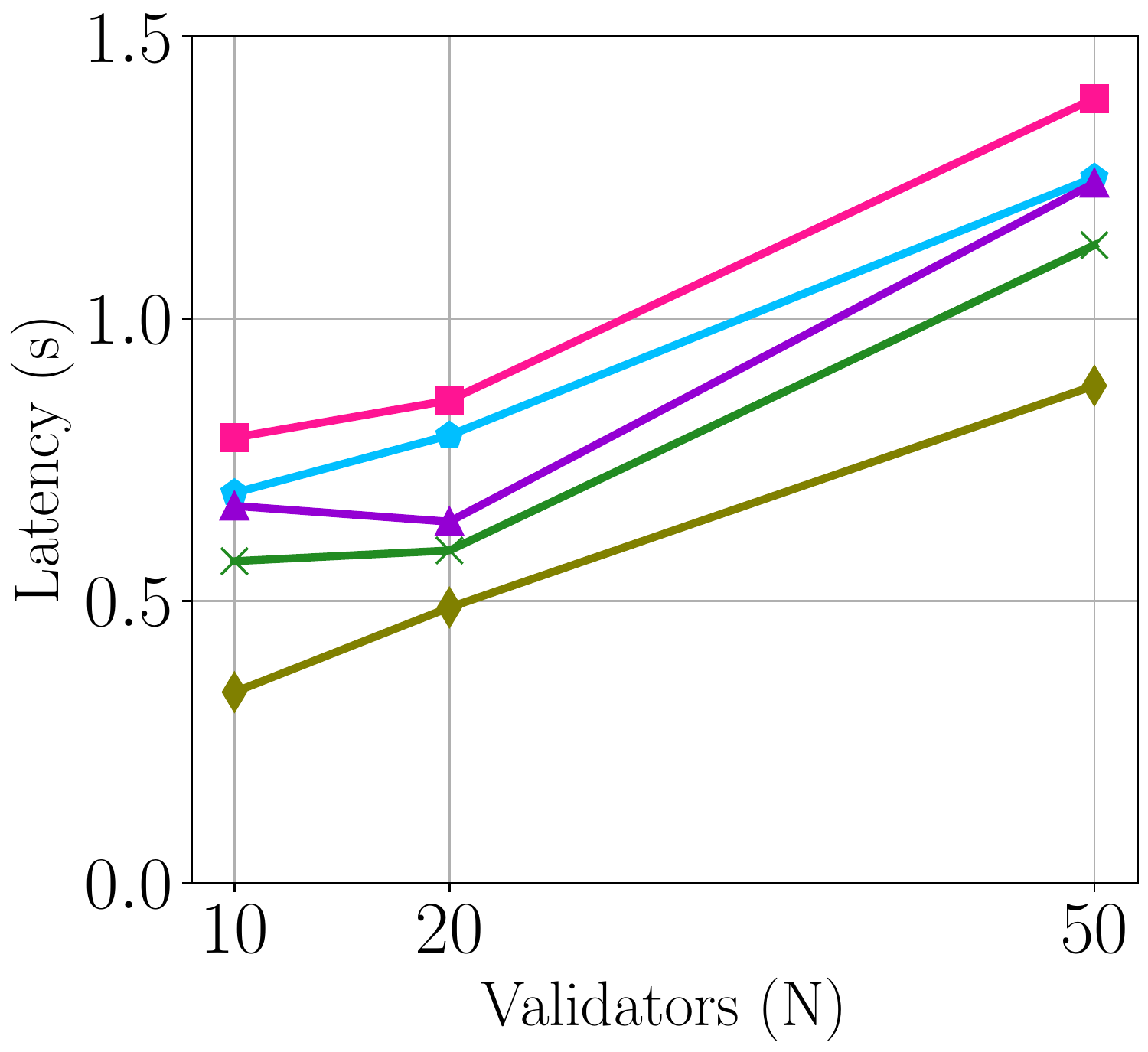}
	\label{fig:shoal-nf-lat}
    \end{subfigure}
    \begin{subfigure}[b]{0.49\columnwidth}
	\centering
	\caption{Vote-round Latency}
	\includegraphics[width=\textwidth]{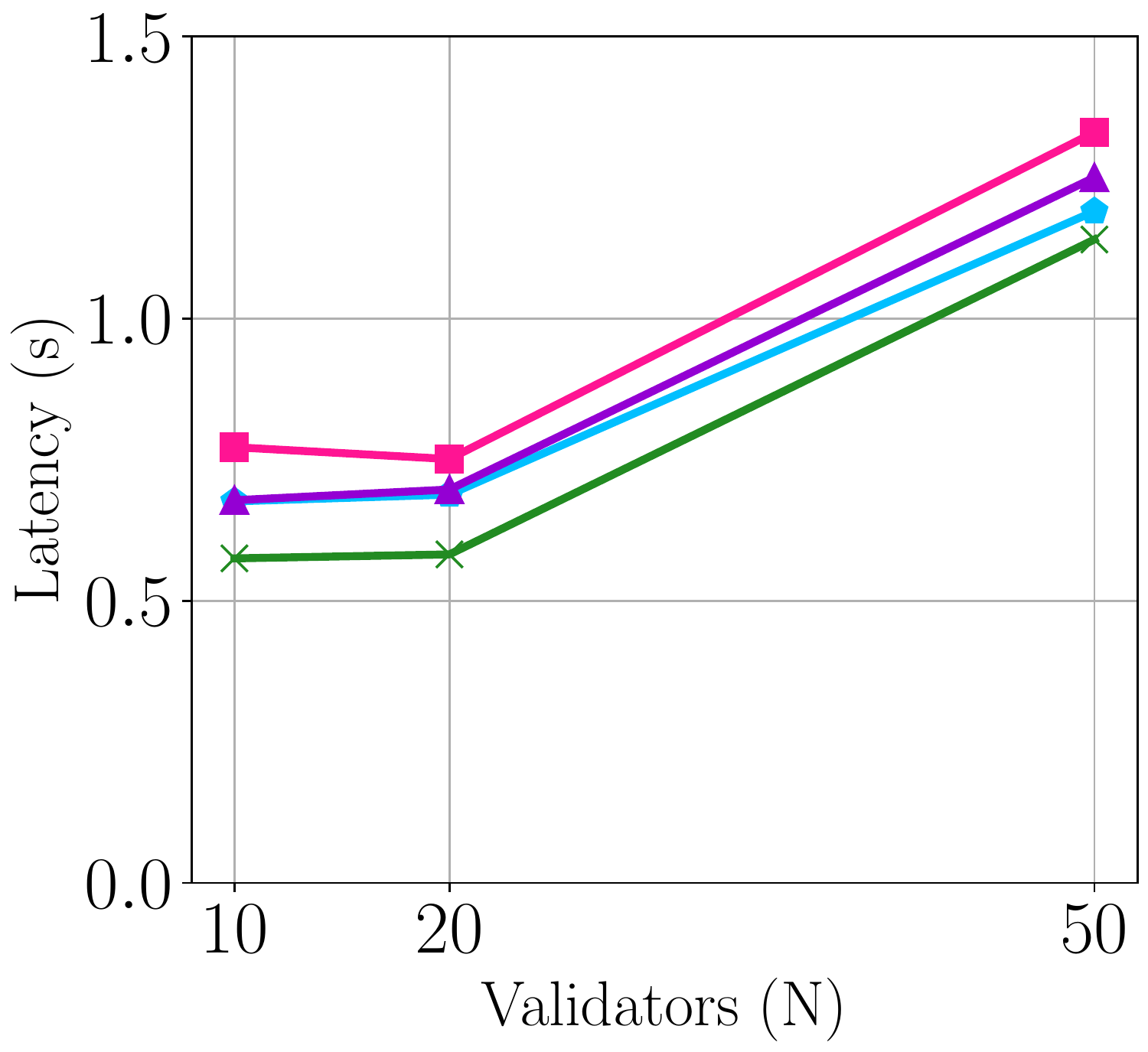}
	\label{fig:shoal-nf-lat-odd}
    \end{subfigure}
    \begin{subfigure}[b]{0.49\columnwidth}
	\centering
	\caption{Anchor-round Latency}
	\includegraphics[width=\textwidth]{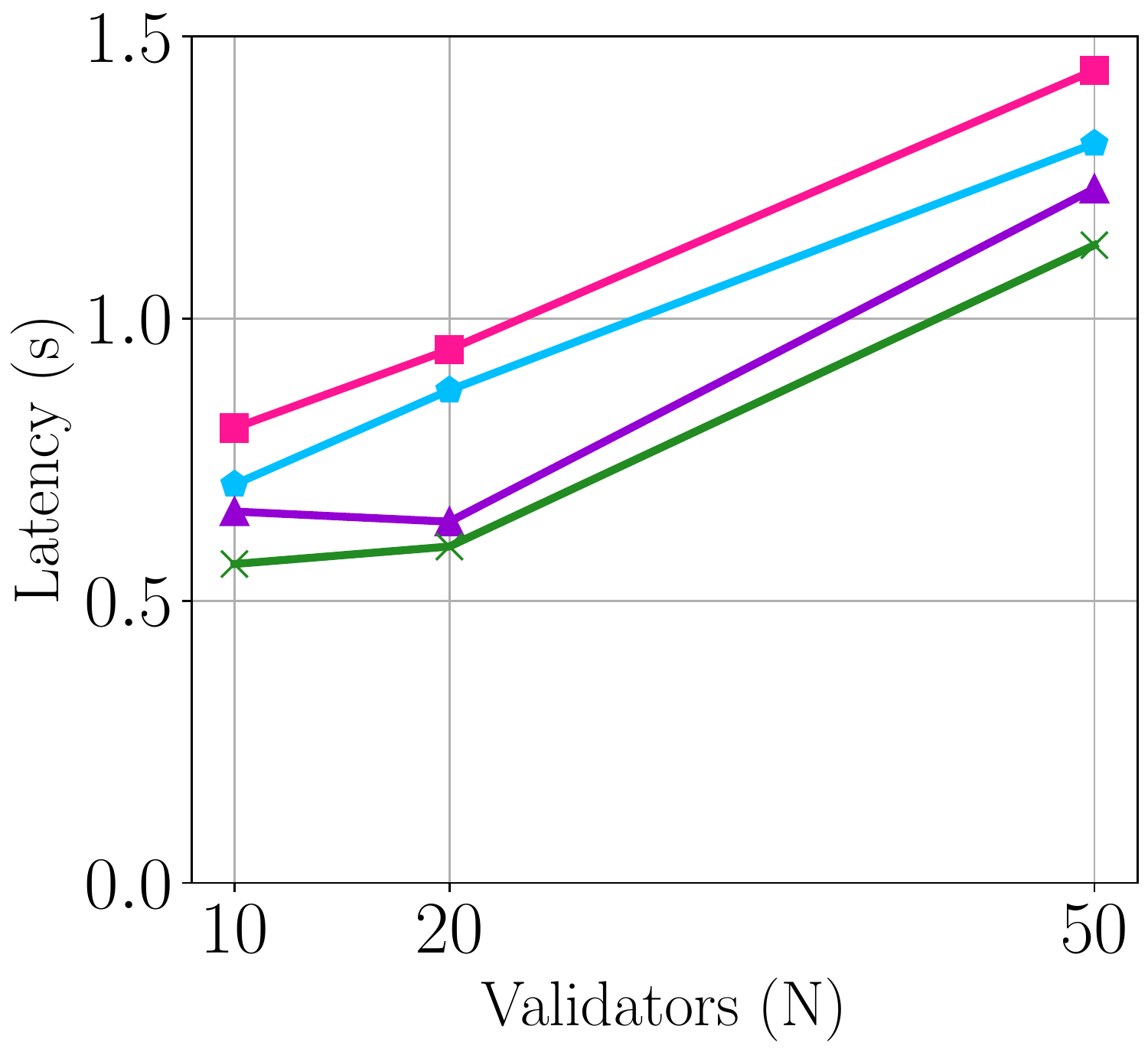}
	\label{fig:shoal-nf-lat-even}
    \end{subfigure}
    \caption{\sysname performance under no failures with 10, 20 and 50 validators.}
    \label{fig:shoal-nf-plots}
\end{figure}

\subsection{Performance of \sysname under fault-free case}

We now evaluate the \sysname variants against the baseline under the normal case where there are no failures. The results are in Figure~\ref{fig:shoal-nf-plots}. As expected, the throughput of the \sysname variants is similar as the number of validators increases. It can be observed that each variant of \sysname decreases the latency leading to full \sysname protocol. In summary, we observe that the \sysname's average latency decreases by up to 20\% compared to Baseline Bullshark.

On the other hand, Jolteon~\cite{jolteon}, despite its use Narwhal's data dissemination decoupling, is only able to achieve a peak throughput of less than 60k, about 40\% lower than \sysname.
This is because under high load leaders become the bottleneck again as they are not able to deal with the required network bandwidth, and as a result, unable to drive progress before timeouts expire.
Furthermore, in terms of latency, Jolteon is $\approx$50\% better than Vanilla Bullshark, but only $\approx$20\% better than \sysname.
Note that the latencies presented do not include the pre-step Quorum Store's latencies, because all the compared protocols include this optimization. However, in the case of \sysname, this latency can be avoided by merging Quorum Store into the DAG construction, as done in Narwhal, which will further close the latency gap from Jolteon.

In Figures~\ref{fig:shoal-f-lat-odd} and~\ref{fig:shoal-f-lat-even}, we distinguish the latencies of transactions in the vote-round vertices from that in anchor-round vertices, in order to show the effect of the pipelining approach. 
The vote and anchor round latencies for \sysname PL, as well as \sysname, are similar, which helps provide predictable and smooth latency for transactions in real production systems. In contrast, the vote and anchor round latencies for Baseline Bullshark and \sysname LR differ by 5-20\% depending on the number of failures.

\subsection{Performance of \sysname under faults}

\begin{figure}[t]
    \begin{subfigure}[b]{0.49\columnwidth}
	\centering
	\caption{Throughput}
	\includegraphics[width=\textwidth]{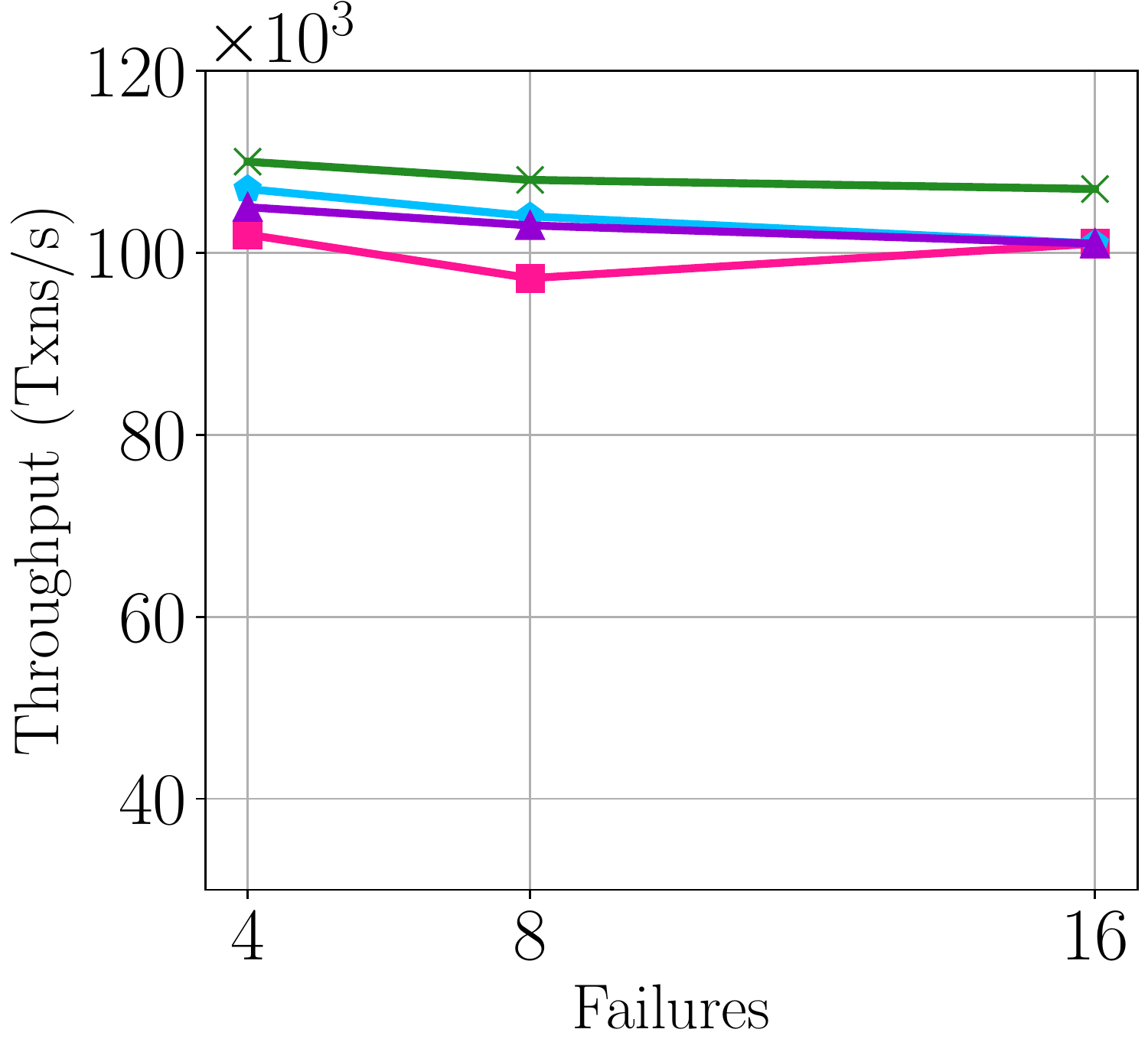}
	\label{fig:shoal-f-tps}
    \end{subfigure}
    \begin{subfigure}[b]{0.49\columnwidth}
        \centering
	\caption{Average Latency}
        \includegraphics[width=\textwidth]{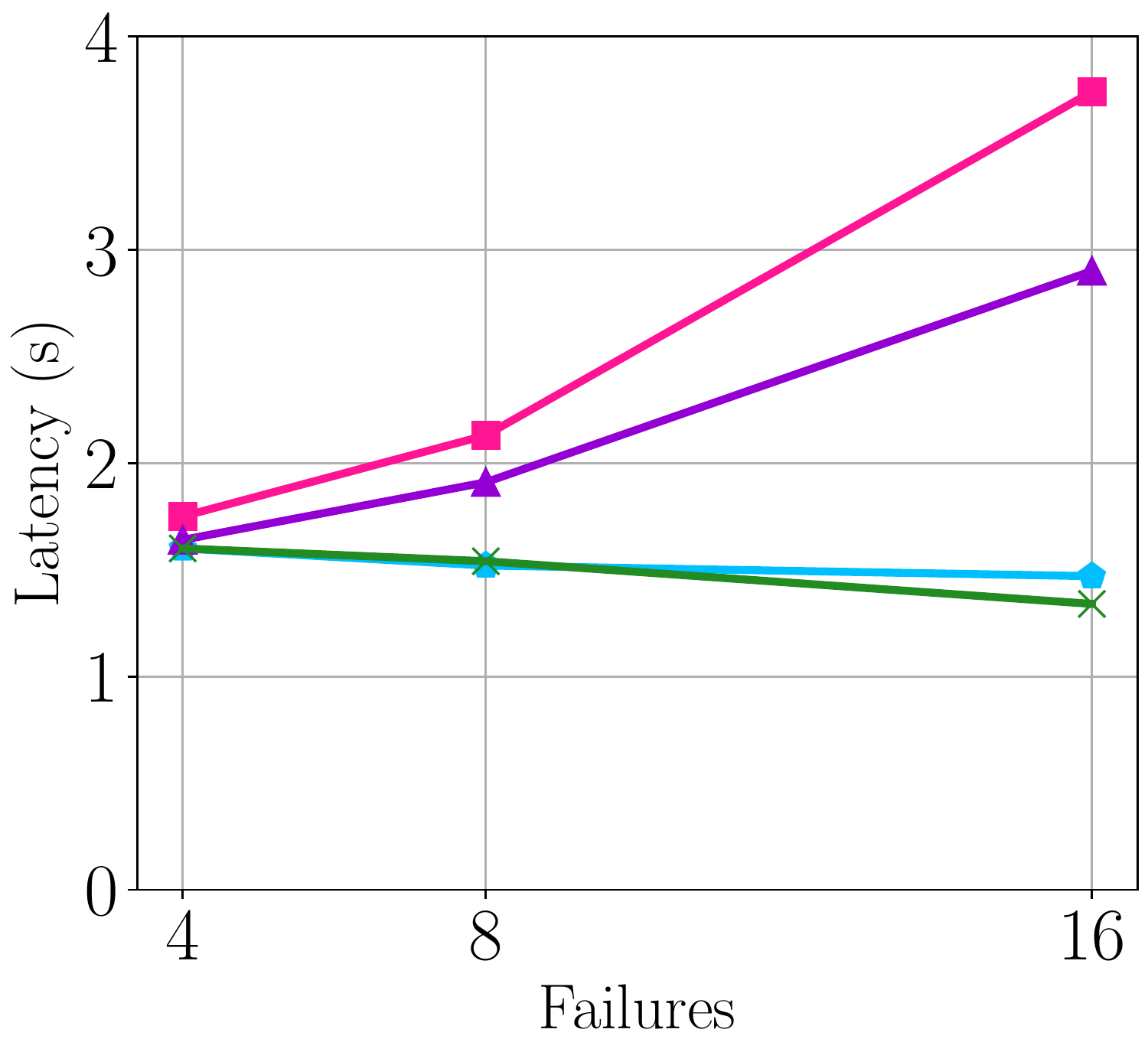}
	\label{fig:shoal-f-lat}
    \end{subfigure}
    \begin{subfigure}[b]{0.49\columnwidth}
	\centering
	\caption{Vote-round Latency}
	\includegraphics[width=\textwidth]{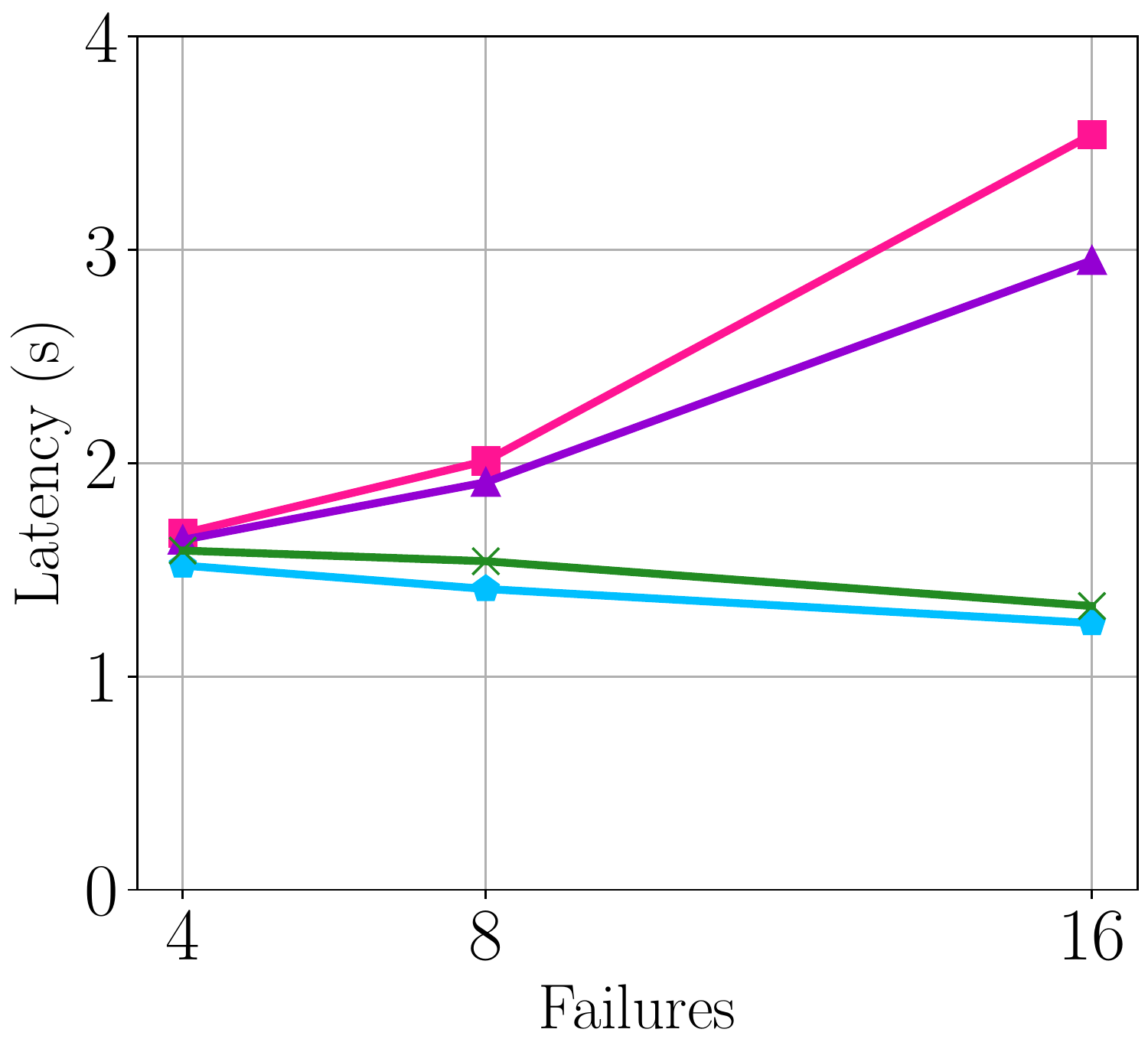}
	\label{fig:shoal-f-lat-odd}
    \end{subfigure}
    \begin{subfigure}[b]{0.49\columnwidth}
	\centering
	\caption{Anchor-round Latency}
	\includegraphics[width=\textwidth]{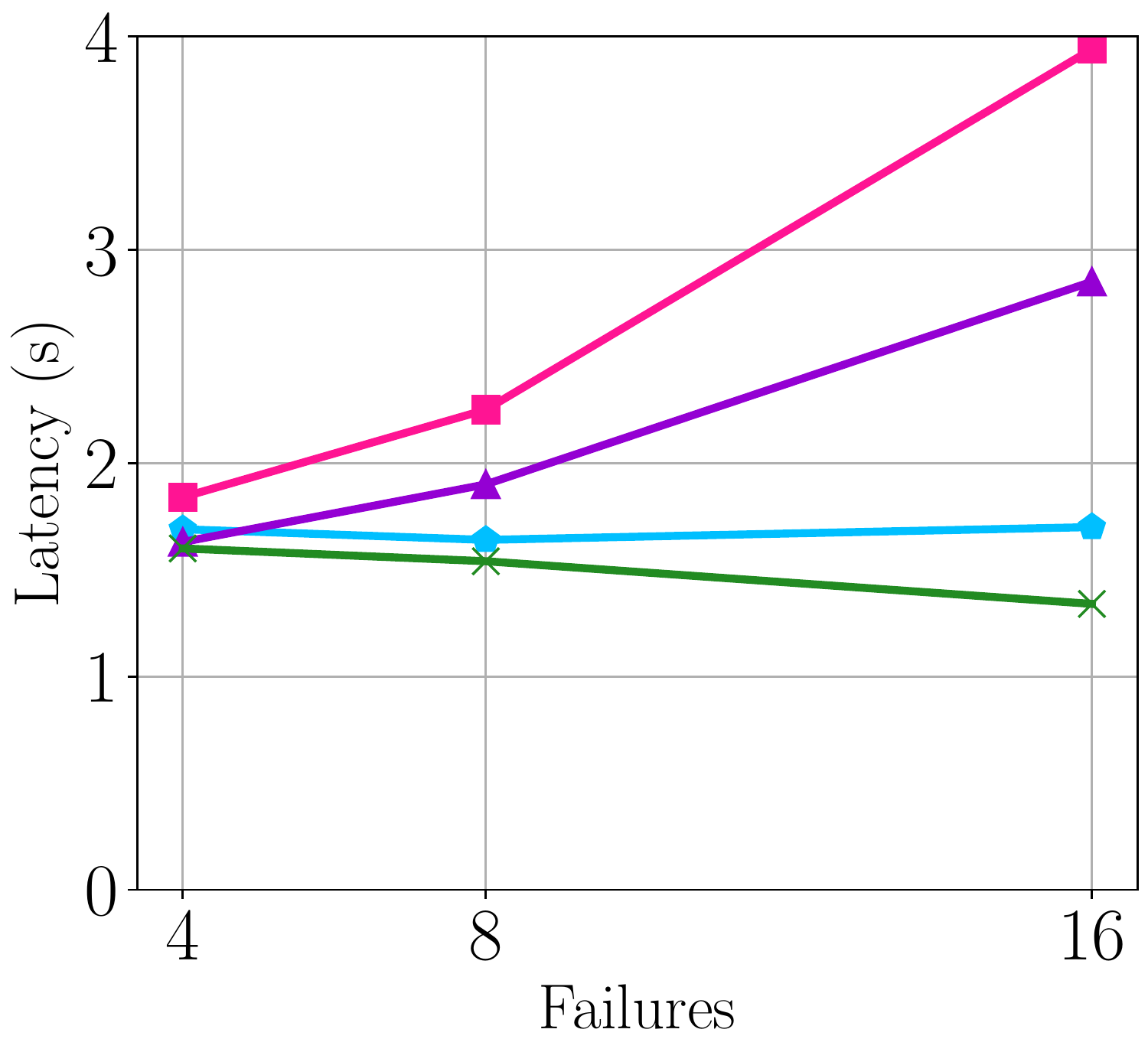}
	\label{fig:shoal-f-lat-even}
    \end{subfigure}
    \caption{\sysname performance under 4, 8, and 16 failures (N=50)}
    \label{fig:shoal-f-plots}
\end{figure}

Figure~\ref{fig:shoal-f-plots} shows the behavior of baseline and \sysname variants under faults. For this experiment, $N=50$ and the failures are increased from 4 to 16 (maximum tolerated). This is the case where the Leader Reputation mechanism helps to improve the latency significantly by reducing the likelihood of failed validators from being anchors.  
Notice that without Leader Reputation, the latencies of Baseline Bullshark and \sysname PL increases significantly as the number of failures increases. \sysname provides up to 65\% lower latencies than Baseline Bullshark under failures.

Figure~\ref{fig:failure-timeline} shows the impact of skipping leaders on the latency by comparing vanilla Bullshark with \sysname on a timeline plot under failures. We have a system of 50 validators, 8 of which have failed. The x-axis represents a part of the experiment time window and the y-axis shows the latency. 
The presence of timeouts and the need to skip anchors causes vanilla Bullshark's latency to fluctuate. In our experiment, we observed latency jitter of approximately one second, which makes it impossible to provide predictable latency in production systems. In constrast, \sysname maintains consistent low latency without any jitter.

\begin{figure}[t]
    \includegraphics[width=0.95\columnwidth]{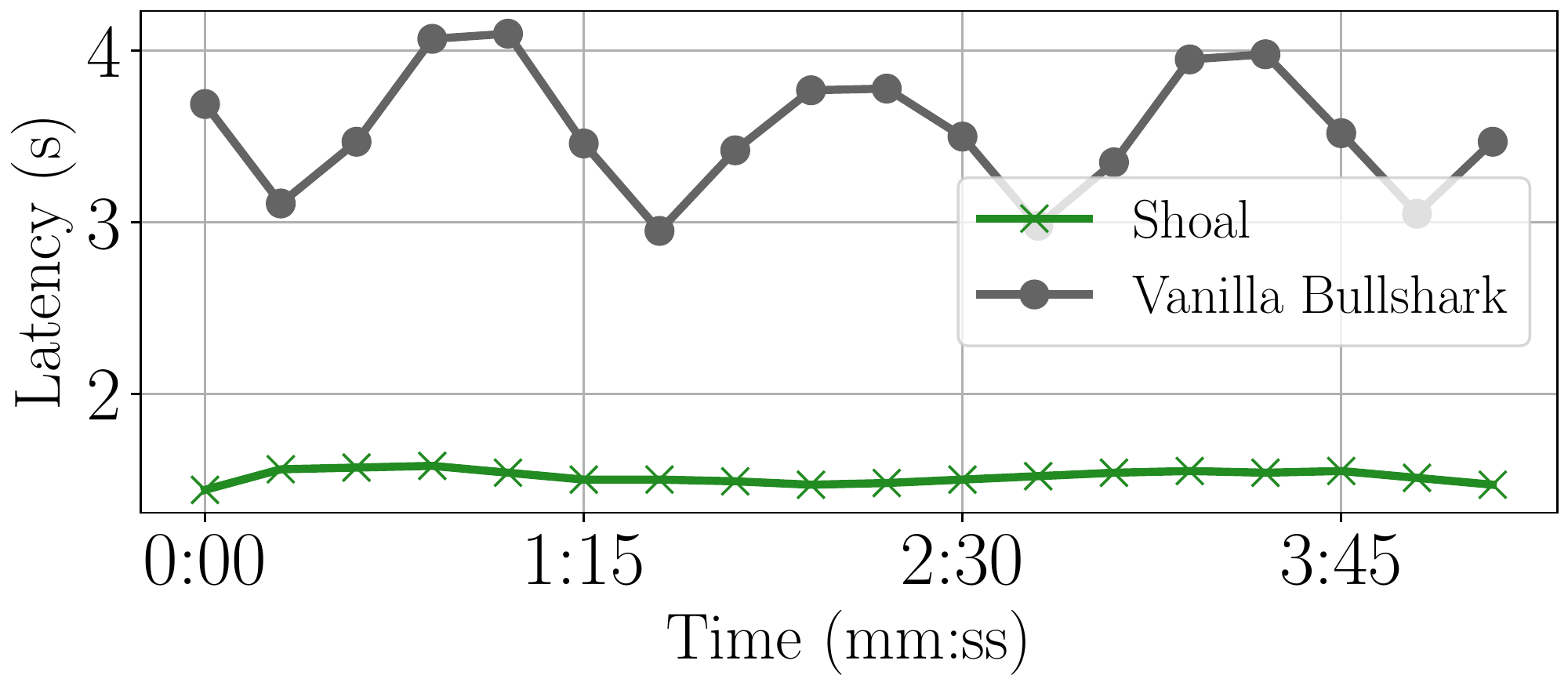}
    \caption{Latency timeline under 8 failures with $N=50$. The x-axis represents a part of the experiment time window and the y-axis shows the latency.}
    \label{fig:failure-timeline}
\end{figure}

\subsection{Summary}

In contrast to Vanilla Bullshark, \sysname provides up to 40\% lower latency in the fault-free case and up to 80\% lower latency under failures. Furthermore, we show that \sysname provides predictable latency and is able to commit at network speed in most cases and without waiting for timeouts.

% TBD

% We evaluated Bullsahrk. 

% Leader reputation: compare performance with crash failures (geo-distribution).

% Pipelining: compare best case (geo-distribution), crash failures (geo-distribution), and random delay case (local forge).

% Combination: crash failures (geo-distribution). Should give us the same numbers as Pipelining's best case.

% explain how timeouts slow us down -- we do not even now that some leaders are slow, but we wait fir them every time. 
\section{Related work} \label{sec:related}

\subsection{BFT systems for Blockchains}

Byzantine fault tolerance (BFT) has been an active area of research for over four decades, with a significant body of literature in both theory~\cite{cachin2011introduction} and systems~\cite{bessani2014state, castro2002practical, DBLP:conf/sosp/Abd-El-MalekGGRW05, zyzzyva}. With the advent of Blockchain systems in recent years, the focus on performance and scalability has notably increased.

Initial efforts to enhance throughput and scalability attempted to reduce the communication complexity of leader-based eventually synchronous protocols. This resulted in a considerable body of work aiming to achieve communication complexity linear to the number of validators~\cite{hotstuff, Lumi, bravo2022making, naor2020expected}. 
Despite sound theoretical premise, the practical implications arguably fell short of expectations.
An independent evaluation and comparison conducted by~\cite{bottlenecks} revealed that the well-known HotStuff~\cite{hotstuff} protocol achieved a throughput of only 3,500 TPS on a geo-replicated network.

The practical breakthrough occurred a few years later with the realization that the main bottleneck in BFT systems, particularly those relying on leaders, is data dissemination. Mir-BFT~\cite{MirBFT} introduced an innovative approach by running multiple PBFT~\cite{castro2002practical} instances in parallel. 
Independently, Narwhal~\cite{narwhaltusk} and later Dispersedledger~\cite{dispersedledger} decoupled data dissemination from the consensus logic. These advancements showcased impressive results, with Narwhal achieving a peak throughput of 160,000 TPS.

There has been systems~\cite{honeybadger, jolteon, beat, dumboccs} and theoretical~\cite{vaba, dumbopodc, cachin2005random} research in asynchronous BFT protocols. However, to the best of our knowledge, no asynchronous protocol is deployed in production in an industrial system. 
Another appealing property of Narwhal is the support of a partially synchronous~\cite{bullsharksync} as well as asynchronous~\cite{narwhaltusk, allyouneed, bullshark} (as long as randomness is available) protocols, and the ability to easily switch among them.

\subsection{Timeouts and responsiveness}
The FLP~\cite{flp} impossibility result states that there is no deterministic consensus protocol that can tolerate a fully asynchronous network.
The proof relies on the fact that it is impossible to distinguish between crashed and slow validators during asynchronous periods.
The immediate application to partially synchronous networks, therefore, is that all deterministic protocols must rely on timeouts in some way to guarantee liveness against a worst-case adversary.
Indeed, to the best of our knowledge, all previous deterministic BFT protocols, including the partially synchronous version of Bullshark~\cite{bullsharksync}, relied on timeouts to implement a simple version of a failure detector~\cite{chandra1996weakest}.
This mechanism monitors the leaders and triggers view-changes when timeouts expire, i.e. when faults are suspected.

The optimistic responsiveness property, popularized by HotStuff~\cite{hotstuff}, avoids timeouts in the best-case failure-free scenario. 
However, when failures do occur, all validators wait until the timeout expires before view-changing to the next leader, introducing a significant slowdown in the protocol execution.
Moreover, as discussed in Section~\ref{sec:implementation}, setting a proper timeout duration is a non-trivial problem in its own right.

\sysname provides prevalent responsiveness, which is a strictly better property than optimistic responsiveness as it guarantees network speed progress in case of healthy leaders and zero delays in case of failures.
\sysname achieves this by relying on the network speed ``clock" inherent in the DAG construction itself~\cite{fordclock}, combined with the leader reputation mechanism.
While due to the FLP result, the worst case in which a timeout would be required for maintaining the liveness of the protocol cannot completely be eliminated, \sysname successfully relegates such cases to occur in specific extremely uncommon scenarios from a practical point of view (multiple consecutive unordered anchors).

\subsection{DAG-based BFT}
DAG-based consensus in the context of BFT was first proposed by HashGraph~\cite{hashgraph}. The idea is to separate the network communication layer, i.e. efficiently constructing a system that forms a DAG of messages, and the consensus logic that can involve complex pieces such as view-change and view-synchronization.
The consensus logic is performed locally, whereby a validator examines its local view of the DAG and orders the vertices without sending any messages.
The challenge arises from the asynchronous nature of the network, which may cause different validators to observe slightly different portions of the DAG. To address this, the DAG structure is interpreted as a consensus protocol, wherein a vertex represents a proposal and an edge represents a vote.

Aleph~\cite{aleph} introduced a round-based DAG structure. Such a structure simplifies support for garbage collection and non-equivocation, which in turn simplifies the consensus logic to order the vertices. 
Narwhal implements round-based DAG, and three Narwhal-based consensus protocols have been previously proposed. The first is DAG-Rider~\cite{allyouneed}, which introduced a quantum-safe asynchronous protocol with optimal amortized communication complexity and $O(1)$ latency. Tusk~\cite{narwhaltusk} improved latency in the best case. An asynchronous version of Bullshark~\cite{bullshark, bullsharksync} includes a fast path~\cite{bullshark}, while a stand-alone partially synchronous protocol~\cite{bullsharksync} also exists and is currently deployed in production in Sui~\cite{sui}. \sysname presents a framework that applies to all Narwhal-based protocols, enhancing their latency through a more efficient ordering rule and a leader reputation mechanism.

An orthogonal theoretical effort~\cite{keidar2022cordial} trades off the non-equivocation property of the DAG construction (which typically requires reliable broadcast), as well as the separation from the consensus logic, in order to reduce latency.

\subsection{Pipelining}
To the best of our knowledge, pipelining in the BFT context was first proposed by Tendermint~\cite{buchman2018latest}, and later utilized in HotStuff~\cite{hotstuff} and Diem~\cite{diembft}. 
State machine replication (SMR) systems can be constructed from multiple instances of single-shot consensus~\cite{lamport2019part}, e.g. one approach to build Byzantine SMR is by running a PBFT instance~\cite{castro2002practical} for each slot.
Tendermint introduced the elegant idea of chaining proposals or piggybacking single-shot instances such that a value for a new slot could be proposed before the value for the previous slot was committed. In this approach, a message in the $i^{th}$ round of the $k^{th}$ instance can be interpreted as a message in round $i-1$ of instance $k+1$. While the latency for each instance remains unchanged, clients experience improved latency as their transactions can be proposed earlier.

In DAG-based consensus, the concept of piggybacking proposals is inherent in the design, as each vertex in the DAG links to vertices in previous rounds. However, previous protocols did not allow having an anchor in every round. 
%resulting in some vertices requiring more rounds to be ordered compared to others, as they are further from an anchor.
\sysname framework supports having an anchor in each round in a good case for any Narwhal-based protocol, providing a "pipelining effect".

\subsection{Leader reputation}
Leader reputation is often overlooked in theory, yet it plays a crucial role in performance in practice. 
While Byzantine failures are rare as validators are highly protected, isolated, and economically incentivized to follow the protocol, more common are validators that are unresponsive.
This may be because they temporarily crashed, running slow hardware, or are simply located farther away.
If a leader/anchor election is done naively, unresponsive validators will unavoidably stall progress and lead to significant performance impact.

A practical approach, implemented in Diem~\cite{diembft} and formalized in~\cite{cohen2022aware}, is to exclude underperforming validators from leader election. This is achieved by updating the set of candidates after every committed block based on the recent activity of validators. In a chained protocol, if all validators observe the same committed block, they can deterministically elect future leaders based on the information in the chain. However, in some cases, certain validators may see a commit certificate for a block earlier than others. This can lead to disagreements among validators regarding the list of next leaders, causing a temporary loss of liveness.

For DAG-based protocols, disagreements on the identity of round leaders can lead the validators to order the DAG completely differently. This poses a challenge for implementing leader reputation on the DAG. As evidence, a Narwhal and Bullshark implementation currently deployed in production in Sui blockchain does not support such a feature~\footnote{github.com/MystenLabs/sui/blob/main/narwhal/consensus/src/bullshark.rs}. \sysname enables leader reputation in Narwhal-based BFT protocols without any additional overhead.

%\balance
\section{Discussion} \label{sec:conclusion}
\sysname can be instantiated with any Narwhal-based consensus protocol, and can even switch between protocols during the DAG retrospective re-interpretation step.

\sysname uniformizes the latency and throughput across the validators and eliminates the use of timeouts except in very rare cases, which contributes to the robustness and performance of the system. 
Predictable and smooth latency and throughput patterns have major practical benefits for real systems. It facilitates setting up effective monitoring and alerts for anomaly detection. This is crucial for ensuring security and quality of service by enabling timely response and any intervention necessary, be it manual or automated. Predictable consensus throughput also facilitates pipelining the ordering of transactions with other components of the Blockchain, e.g. transaction execution and commit.

\sysname satisfies the property we name prevalent responsiveness, ensuring the worst-case executions that must use timeouts due to the FLP impossibility result are aligned with the improbable (and worst-case) scenarios from the practical standpoint. Moreover, the design without timeouts plays into the strengths of the leader reputation mechanism of \sysname, and as a result, provides further latency improvements.

\ifdefined\cameraReady
%\section*{Acknowledgments}

\fi

\bibliographystyle{ACM-Reference-Format}
\bibliography{references}

\appendix
\section{Multiple Anchors per Round}
\label{sec:multi}
With pipelining, \sysname introduces an anchor in every round. 
As a result, in the best case, each anchor requires 2 rounds to commit and while non-anchor vertices require 3 rounds.
Next, we present an approach to further optimize the latency for non-anchor vertices, which relies on retrospectively re-interpreting the DAG structure.

We could envision a protocol in which we iterate over more than one vertex in each round in a deterministic order and treat each vertex as an anchor.
More specifically, for a vertex $v$ in round $r$, we consider executing an instance of the underlying Narwhal-based consensus protocol $\mathcal{P}$ (i.e., DAG-Rider, Tusk, and Bullshark) starting from round $r$ with $v$ being the first anchor.
This involves re-interpreting the existing DAG structure, and potentially letting it evolve, until a decision of whether $v$ is ordered or skipped is locally made.
If $v$ is ordered by $\mathcal{P}$, then the causal history of $v$ followed by $v$ is added to the ordering determined by the new protocol. Otherwise, $v$ is skipped and the protocol proceeds to considering a new instantiation of $\mathcal{P}$ from the next potential anchor (which may be in the same round).

A pseudocode in which all vertices are considered as anchors appears in Algorithm~\ref{alg:bounus}.
\begin{algorithm}
\caption{Every vertex as an Anchor}
\begin{algorithmic}[1]
        \State $r \gets 0$
        \While{\emph{true}}
            \For{\textbf{each} validator $k$}
                \State let $v_{r,k}$ be a vertex by validator $k$ in round $r$
                \State let $F_{r,k}: R \rightarrow A$ be a known to all mapping from 
                \Statex \hspace{9.5mm} rounds to anchors such that $F_{r,k}(r) = k$
                \State execute $\mathcal{P}$, select anchors by $F_{r,k}$, starting from 
                \StateX\hspace{4.5mm} $r$ until the first ordered (not skipped) anchor $A$ 
                \StateX\hspace{4.5mm} is determined.
                \If{$A = v_{r,k}$}
                    \State order $A$'s causal history according to $\mathcal{P}$ 
                \EndIf
            \EndFor
        \State $r \gets r+1$ 

        \EndWhile
    
\end{algorithmic}
\label{alg:bounus} 
\end{algorithm}

In the good case, each vertex that is considered as an anchor can be ordered in 2 rounds. 
However, the drawback of this approach is that if some validators are slow and a potential anchor takes many rounds to decide whether to skip or order, the progress of the whole protocol will be stalled.
This happens because potential anchor vertices must be considered in an agreed-upon and deterministic order. As a result, a vertex that necessitates more rounds incurs a latency penalty for the subsequent vertices.

The above issue can potentially be mitigated by combining it with a leader reputation mechanism to select the vertices that are considered as potential anchors, making the bad case delays less likely.
The other vertices can be ordered based on causal history as previously.

\ifdefined\cameraReady

\fi

\end{document}